\documentclass[10pt,a4paper]{article}
\usepackage[utf8]{inputenc}
\usepackage[english]{babel}
\usepackage{amsmath}
\usepackage{amsthm}
\usepackage[hidelinks,colorlinks,allcolors=blue]{hyperref}
\usepackage{amsfonts}
\usepackage{amssymb}
\usepackage{graphicx}
\usepackage{float}
\usepackage{enumerate}
\usepackage{bbm}
\usepackage{xcolor}
\usepackage[left=2cm,right=2cm,top=2cm,bottom=2cm]{geometry}
\usepackage[round]{natbib}   
\author{Augustin Touron}
\title{Consistency of the maximum likelihood estimator in seasonal hidden Markov models}

\newcommand{\R}{\mathbb{R}}

\newtheorem{hyp}{}

\newtheorem{theo}{Theorem}
\newtheorem{lem}{Lemma}

\DeclareMathOperator*{\argmax}{arg\,max}

\begin{document}

\maketitle

\paragraph{Abstract}In this paper, we introduce a variant of hidden Markov models in which the transition probabilities between the states, as well as the emission distributions, are not constant in time but vary in a periodic manner. This class of models, that we will call seasonal hidden Markov models (SHMM) is particularly useful in practice, as many applications involve a seasonal behaviour. However, up to now, there is no theoretical result regarding this kind of model. We show that under mild assumptions, SHMM are identifiable: we can identify the transition matrices and the emission distributions from the joint distribution of the observations on a period, up to state labelling. We also give sufficient conditions for the strong consistency of the maximum likelihood estimator (MLE). These results are applied to simulated data, using the EM algorithm to compute the MLE. Finally, we show how SHMM can be used in real world applications by applying our model to precipitation data, with mixtures of exponential distributions as emission distributions.

\section{Introduction}

Hidden Markov models (HMM) have had various application fields during the last decades: finance \citep{mamon2007}, ecology \citep{patterson2017}, climate modelling \citep{wilks98}, speech recognition \citep{gales2008}, genomics \citep{yoon2009} and many more. Let $\mathsf{X}$ be a finite set and $(\mathsf{Y},\mathcal{Y})$ a measurable space. A \emph{hidden Markov model} (HMM) with state space $\mathsf{X}$ is a $\mathsf{X}\times\mathsf{Y}$-valued stochastic process $(X_t,Y_t)_{t\geq 1}$ where $(X_t)_{t\geq 1}$ is a Markov chain and $(Y_t)_{t\geq 1}$ are $\mathsf{Y}$-valued random variables that are independent conditonnally on $(X_t)_{t\geq 1}$ and such that for all $j\geq 1$, the conditionnal distribution of $Y_j$ given $(X_t)_{t\geq 1}$ only depends on $X_j$. The law of the Markov chain $(X_t)_{t\geq 1}$ is determined by its initial distribution $\pi$ and its transition matrix $\mathbf{Q}$. For all $k\in\mathsf{X}$, the distribution of $Y_1$ given $X_1=k$ is called the \emph{emission distribution} in state $k$. The Markov chain $(X_t)_{t\geq 1}$ is called the \emph{hidden} Markov chain because it is not accessible to observation. The process $(Y_t)_{t\geq 1}$ only is observed. See \citep{rabiner86} for an introduction to HMM and \cite{cappe2009} for a more general formulation. One very common approach to fit such models is to give a parametric form to the emission distributions and to infer the parameters by maximizing the likelihood function. The asymptotic properties of such an estimator have been widely studied. In \cite{baum1966}, the authors proved the consistency and the asymptotic normality of the maximum likelihood estimator (MLE) when the emission distributions have a finite support. Since then, these results have been extended to more general hidden Markov models: see e.g. \cite{douc2004} and references therein.

\paragraph{Motivation} In many applications, especially in climate modeling, simple stationary HMM are not adapted because the data exhibit non-stationarities such as trends and seasonal behaviours. Obviously, temperature has a seasonal component, but this is also the case of precipitations or wind speed for example. It is sometimes possible to preprocess the data in order to obtain a stationary residual. However, finding the right form for the non-stationarity can be very tricky, as well as testing the stationarity of the residual. In the case where the non-stationarity is caused by the presence of a seasonality, a popular solution to avoid this pitfall is to split the period into several sub-periods, and to assume stationarity over each sub-period. For example, if there is an annual cycle, one may consider each month separately and fit twelve different models. However, this is not entirely satisfactory, for several reasons.
\begin{itemize}
\item A choice has to be made for the length of the time blocks.
\item The stationarity assumption over each sub-period may not be satisfied.
\item We have to fit independently several sub-models, which requires a lot of data.
\item The time series used to fit each of the sub-models is obtained by concatenation of data that do not belong to the same year. For exemple, if the time periods are months, the 31st of January of year $n$ will be followed by the first of January of year $n+1$. This is a problem if we use a Markovian model, which exhibits time dependence.
\item If the purpose is simulation, it is preferable to be able to simulate a full year using only one model.
\end{itemize}
Therefore, we prefer to use a extension of HMM that allows seasonality, both in the transition probabilities and in the emission distributions.  This will be refered to as a Seasonal Hidden Markov Model (SHMM). 

\paragraph{Our contribution}
In this paper, we will first detail the mathematical framework of a general SHMM. Although this generalization of hidden Markov models is very useful in practice, as far as we know there exists no theoritical result concerning this class of models. The first question that arises is the identifiability of such models. The identifiability of stationary hidden Markov models is not obvious and has been solved only recently. In \cite{gassiat2016}, it is proved that the transition matrix and the emission distributions of a hidden Markov model are identifiable from the joint distribution of three consecutive observations, provided that the transition matrix is non-singular and that the emission distributions are linearly independent. \cite{alexandrovich2016} proved that the identifiability can be obtained with the weaker assumption that the emission distributions are  distinct. In this paper, we extend the result of \cite{gassiat2016} by proving that the SHMM are identifiable up to state labelling, under similar assumptions. To achieve this, we use a spectral method as described in \citep{hsu2012}. Once we have proved that SHMM are identifiable, it is natural to seek a consistent estimator for their parameters. Regarding HMM, it has been proved by \cite{douc2011} that under weak assumptions, the maximum likelihood estimator in the framework of HMM is strongly consistent. In this paper, we generalize this result to SHMM. This is done by applying the result of Douc et al. to a well chosen hidden Markov model. We then give several examples of specific models for which our consistency result can be applied. The practical computation of the maximum likelihood estimator for HMM, and a fortiori for SHMM, is not straightforward. We address this problem by adapting the classical EM algorithm to our framework. Then we run this algorithm on simulated data to illustrate the convergence of the MLE to the true parameters. Finally, we successfully fit a SHMM to precipitation data and show that our model is able to reproduce the statistical behaviour of precipitation.
 
\paragraph{Outline}In Section \ref{sec2}, we give the general formulation of a SHMM in a parametric framework. Using a spectral method \citep{hsu2012} and a general identifiability result for HMM \citep{gassiat2016}, we show that under weak assumptions, SHMM are identifiable up to state labelling. Then, we generalize the existing results on the convergence of the MLE in HMM to prove our main result, that is the strong consistency of the MLE in SHMM. We also give two examples of models (using mixtures of exponential distributions and mixtures of Gaussian distributions) for which this convergence theorem is applicable. In Section \ref{sec3}, we describe the EM algorithm used for the numerical computation of the maximum likelihood estimator and we illustrate the convergence of the MLE for a simple SHMM using simulated data. Finally, in Section \ref{apprain}, we fit a SHMM to precipitation data and we show that such a model can be used to simulate realistic times series of weather variables.

\section{Consistency result}\label{sec2}

\subsection{Model description}\label{description}

Let $K$ a positive integer and $(X_t)_{t\geq 1}$ a non-homogeneous Markov chain with state space $\mathsf{X}=\{1,\dots,K\}$ and initial distribution $\pi$, that is $\pi_k = \mathbb{P}(X_1=k)$. For $t\geq 1$ and $1\leq i,j\leq K$, let
$$Q_{ij}(t):=\mathbb{P}(X_{t+1}=j\mid X_t = i).$$
$Q(t)$ is the transition matrix from $X_t$ to $X_{t+1}$. Let us assume that $Q(\cdot)$ is a $T$-periodic function, so that there exists some integer $T\geq 1$ such that for all $t\geq 1$, $Q(t+T) = Q(t)$. Let $\mathsf{Y}$ a Polish space, $\mathcal{Y}$ its Borel $\sigma$-algebra, and $(Y_t)_{t\geq 1}$ a $\mathsf{Y}$-valued stochastic process. Assume that conditionally to $(X_t)_{t\geq 1}$, the $(Y_t)_{t\geq 1}$ are independent and that the distribution of $Y_s$ conditionally to the $(X_t)_{t\geq 1}$ only depends on $X_s$ and $s$. We shall denote by $\nu_{k,t}$ the distribution of $Y_t$ given $X_t=k$. We also assume that for all $t\geq 1$ and for all $k\in\mathsf{X}$, $\nu_{k,t+T}=\nu_{k,t}$. Then the process $(X_t,Y_t)_{t\geq 1}$ is called a \emph{seasonal hidden Markov model} (SHMM) and the $\nu_{k,t}$'s are its \emph{emission distributions}. 

\paragraph{}The law of the process $(X_t,Y_t)_{t\geq 1}$ is determined by the distribution $\pi$ of $X_1$, the transition matrices $Q(1),\dots ,Q(T)$ and the emission distributions $\nu_{k,1},\dots, \nu_{k,T}$ for $1\leq k\leq K$.

\paragraph{}Choosing $T=1$, we retrieve the classical hidden Markov model (HMM). A general SHMM includes periodicity in both the transition probabilitites and the emission distributions. However, for some applications, it may be enough to consider periodic transitions and constant emission distributions, or vice versa (see e.g. Section \ref{apprain}).

\paragraph{}The following remark will be the key to the proof of our main result. Given $(X_t,Y_t)_{t\geq 1}$ a seasonal hidden Markov model, two (classical) hidden Markov models naturally arise.
\begin{itemize}
\item For any $t\in\{1,\dots,T\}$ the process $(X_{jT+t},Y_{jT+t})_{j\geq 0}\in\left(\mathsf{X}\times\mathsf{Y}\right)^\mathbb{N}$ is a hidden Markov model with transition matrix $Q(t)Q(t+1)\dots Q(T)Q(1)\dots Q(t-1)$ and emission distributions $(\nu_{k,t})_{k\in\mathsf{X}}$.
\item For $j\geq 0$, let $U_j:=(X_{jT+1},\dots,X_{jT+T})$ and $W_j:=(Y_{jT+1},\dots,Y_{jT+T})$. Then $(U_j,W_j)_{j\geq 0}\in\left(\mathsf{X}^T\times\mathsf{Y}^T\right)^\mathbb{N}$ is a hidden Markov model. For $u=(u_1,\dots,u_T)\in\mathsf{X}^T$, the conditional distribution of $W_0$ given $U_0=u$ is $\bigotimes_{t=1}^T\nu_{u_t,t}$. The transition matrix of the homogeneous Markov chain $(U_j)_{j\geq 0}$ is given by $$\tilde{Q}_{uv}:=\mathbb{P}(U_1=v\mid U_0 = u)=Q_{u_Tv_1}(T)Q_{v_1v_2}(1)\dots Q_{v_{T-1}v_T}(T-1).$$
\end{itemize}

\paragraph{Parametric framework} Assume that there exists a measure $\mu$ defined on $\mathcal{Y}$ such that all the emission distributions are absolutely continuous with respect to $\mu$ and let $f_{k,t}:=\frac{d\nu_{k,t}}{d\mu}$ be the \emph{emission densities}. We consider that all the emission distributions and the transition matrices depend on a parameter $\theta\in\Theta$, where $\Theta$ is a compact subset of some finite-dimensional vector space, e.g. $\mathbb{R}^q$. Let us precise the structure of $\Theta$.
\begin{itemize}
\item The function $t\mapsto Q(t)$ belongs to a (known) parametric family of $T$-periodic functions indexed by a parameter $\beta$ that we wish to estimate.\paragraph{Example}
\begin{equation}\label{eq-Q}
Q_{ij}(t)\propto \exp\left[\sum_{l=0}^d\left(a_{ijl}\cos\left(\frac{2\pi lt}{T}\right)+b_{ijl}\sin\left(\frac{2\pi lt}{T}\right)\right)\right]
\end{equation}
In this example, $\beta=(a_{ijl},b_{ijl})_{1\leq i,j\leq K,0\leq l\leq d}$ and all the transition matrices are entirely determined by $\beta$.
\item For any $t\geq 1$, the emission densities $f_{k,t}$ belong to a (known) parametric family (which does not depend on $t$) indexed by a parameter $\theta^Y(t)$. In addition, we assume that the $T$-periodic function $t\mapsto\theta^Y(t)$ itself belongs to a parametric family indexed by a parameter $\delta$.
\paragraph{Example} 
Denoting by $\mathcal{E}(\alpha)$ the exponential distribution with parameter $\alpha$, 
$$\nu_{k,t}= \mathcal{E}\left[\delta_k\left(1+\cos\left(\frac{2\pi t}{T}\right)\right)\right]$$
In such a case, $\theta^Y(t)=\left(\delta_k\left(1+\cos\left(\frac{2\pi t}{T}\right)\right)\right)_{k\in\mathsf{X}}$ and $\delta = (\delta_1,\dots,\delta_K)\in\mathbb{R}^K$.
\item[$\bullet$] Hence we can define $\theta = (\beta,\delta)$.
\end{itemize}
We shall denote by $\mathbb{P}^{\pi,\theta}$ the law of the process $(Y_t)_{t\geq 1}$ when the parameter is $\theta$ and the distribution of $X_1$ is $\pi$, and by $\mathbb{E}^{\pi,\theta}(\cdot)$ the corresponding expected value.  If the initial distribution $\pi$ is the stationary distribution associated with the transition matrix $Q(1)\cdots Q(T)$, then the two HMM described above are stationary. In such a case, we will simply write $\mathbb{P}^\theta$ and $\mathbb{E}^\theta(\cdot)$ for the law of $(Y_t)_{t\geq 1}$ and the corresponding expected value. Our purpose is to infer $\theta$ from a vector of observations $(Y_1,\dots,Y_n)\in\mathsf{Y}^n$. We assume that the model is well specified, which means that there exists a \emph{true} initial distribution $\pi^*$ and a \emph{true} parameter $\theta^*=(\beta^*,\delta^*)$ in the interior of $\Theta$ such that the observed vector $(Y_1,\dots,Y_n)$ is generated by the SHMM defined by $\pi^*$ and $\theta^*$. We denote by $Q^*(t)$ and $\nu^*_{k,t}$ the corresponding transition matrices and emission distributions. Note that we consider the number $K$ of hidden states to be known.

\subsection{Identifiability}\label{ident}

In \cite{gassiat2016}, the authors prove that the transition matrix and the emission distributions of a stationary HMM are identifiable (up to state labelling) from the law of three consecutive observations, provided that the transition matrix has full rank and that the emission distributions are linearly independent. Still in the context of HMM, the authors of \citet{alexandrovich2016} use the weaker assumption that  the emission distributions are all distinct to obtain identifiability up to state labelling. However, it requires to consider the law of more than three consecutive observations. In this paragraph, we show that under similar assumptions, the SHMM described above is identifiable: we can retrieve the transition matrices and the emission distributions from the law of the process $(Y_t)_{t\geq 1}$ up to permutations of the states. We will use the following assumptions:

\begin{hyp}\label{hyp-inv}
For $1\leq t\leq T$, the transition matrix $Q^*(t)$ is invertible and irreducible.
\end{hyp}
\begin{hyp}\label{hyp-stat}
The matrix $Q^*(1)\cdots Q^*(T)$ is ergodic and its unique stationary distribution $\pi^*$ is the distribution of $X_1$. 
\end{hyp}
\begin{hyp}\label{hyp-ind} 
For $1\leq t\leq T$, the $K$ emission distributions $(\nu^*_{k,t})_{k\in\mathsf{X}}$ are linearly independent.
\end{hyp}

\paragraph{Remark}
These assumptions only involve the data generating process and its corresponding parameter $\theta^*$, they are not about the whole parameter space $\Theta$.

\begin{theo}\label{theo-id}
Assume that the set of true parameters $(\pi^*,Q^*(t),\nu^*_{k,t})_{1\leq t\leq T,k\in\mathsf{X}}$ satisfies Assumptions \ref{hyp-inv}-\ref{hyp-ind}. Let $(\tilde{\pi},\tilde{Q}(t),\tilde{\nu}_{k,t})_{1\leq t\leq T,k\in\mathsf{X}}$ another set of parameters such that the joint distribution of $(Y_1,\dots, Y_T)$ is the same under both sets of parameters. Then there exist $\sigma_1,\dots,\sigma_T$ permutations of $\mathsf{X}$ such that for all $k\in\mathsf{X}$, $\tilde{\pi}_k = \pi_{\sigma_1(k)}$ and for all $t\in\{1,\dots, T\}$, $k,l\in\mathsf{X}$, $\tilde{\nu}_{k,t}=\nu_{\sigma_t(k),t}$ and $\tilde{Q}(t)_{kl} = Q(t)_{\sigma_t(k),\sigma_{t+1}(l)}$, with $\sigma_{T+1}=\sigma_1$.
\end{theo}

\begin{proof}We shall follow the spectral method as presented in \cite{hsu2012} (see also \cite{DCGL2016} and \cite{DCGLC2017}). Before going through the spectral algorithm, let us first introduce some notations. Let $(\phi_n)_{n\in\mathbb{N}}$ be a sequence of measurable functions such that for any probability measures $\nu_1,\nu_2$ on $(\mathsf{Y},\mathcal{Y})$, 
$$\left(\int_\mathsf{Y}\phi_nd\nu_1\right)_{n\in\mathbb{N}} = \left(\int_\mathsf{Y}\phi_nd\nu_2\right)_{n\in\mathbb{N}}\implies \nu_1=\nu_2.$$ Such a sequence exists because $\mathsf{Y}$ is a Polish space.
For $t\geq 2$ and $N\geq 1$, we shall consider the following matrices:

\begin{itemize}
\item Let $O_t = O_t^{(N)}\in\mathbb{R}^{N\times K}$ be the matrix defined by 
$$(O_t)_{ak}=\mathbb{E}\left[\phi_a(Y_t)\mid X_t = k\right]=\int_Y\phi_ad\nu^*_{k,t}.$$ 
\item Let $L(t)\in\mathbb{R}^N$ be the vector such that, for $1\leq a\leq N$, $L_a(t) = \mathbb{E}\left[\phi_a(Y_t)\right]$.
\item Let $N(t)\in\mathbb{R}^{N\times N}$ be the matrix defined by $N_{ab}(t)=\mathbb{E}\left[\phi_a(Y_t)\phi_b(Y_{t+1})\right]$.
\item Let $P(t)\in\mathbb{R}^{N\times N}$ be the matrix defined by $P_{ac}(t)=\mathbb{E}\left[\phi_a(Y_{t-1})\phi_c(Y_{t+1})\right]$.
\item For $b\in\{1,\dots,N\}$, let $M_t(\cdot,b,\cdot)$ be the matrix defined by $M_t(a,b,c) = \mathbb{E}[\phi_a(Y_{t-1})\phi_b(Y_t)\phi_c(Y_{t+1})]$.
\end{itemize}

Notice that all these quantities can be computed from the law of $(Y_t)_{1\leq t\leq T}$, except $O_t$ which requires the emission distributions. Although they all depend on $N$, we do not indicate it, for the sake of readability. Using Assumption \ref{hyp-ind}, we see that there exists an integer $N_0>K$ such that for all $N\geq N_0$, the matrices $O_t^{(N)}$ have full rank. From now on, we will consider that $N\geq N_0$. We denote by $\pi^*(t)$ the (unconditional) distribution of $X_t$. Elementary calculations show that the following equalities hold:

\begin{align}
L(t)&=O_t\pi^*(t)\label{eq1}\\
N(t)&= O_t\mathrm{diag}(\pi^*(t))Q^*(t)O_{t+1}^T\label{eq2}\\
P(t)&= O_{t-1}\mathrm{diag}(\pi^*(t-1))Q^*(t-1)Q^*(t)O_{t+1}^T\label{eq3}\\
M_t(\cdot,b,\cdot) &= O_{t-1}\mathrm{diag}(\pi^*(t-1))Q^*(t-1)\mathrm{diag}[O_t(b,\cdot)]Q^*(t)O_{t+1}^T,\quad 1\leq b\leq N\label{eq4}
\end{align}

where $\mathrm{diag}(v)$ is the diagonal matrix whose diagonal entries are those of the vector $v$. Thanks to Assumptions \ref{hyp-inv}-\ref{hyp-stat}, all the entries of $\pi^*(t)$ are positive, so that $\mathrm{diag}(\pi^*(t))$ is invertible. In addition, Assumption \ref{hyp-inv} and equations \eqref{eq2} and \eqref{eq3} show that the matrices $P(t)$ and $N(t)$ also have rank $K$.

\paragraph{}Let $P(t) = U\Sigma V^T$ be a singular value decomposition (SVD) of $P(t)$: $U$  and $V$ are matrices of size $N\times K$ whose columns are orthonormal families being the left (resp. right) singular vectors of $P(t)$ associated with its $K$ non-zero singular values, and $\Sigma=U^TP(t)V$ is an invertible diagonal matrix of size $K$ containing these singular values. Note that such a decomposition is not unique, as we may choose arbitrarily the order of the diagonal entries of $\Sigma$, which is equivalent to swapping the columns of $U$ and $V$, using the same permutation of $\{1,\dots,K\}$. Let us define, for $1\leq b\leq N$: 
$$B(b):=(U^TP(t)V)^{-1}U^TM_t(\cdot,b,\cdot)V.$$

Hence we have:

\begin{align*}
B(b) &= \left(U^TO_{t-1}\mathrm{diag}(\pi^*(t-1))Q^*(t-1)Q^*(t)O_{t+1}^TV\right)^{-1}U^TO_{t-1}\mathrm{diag}(\pi^*(t-1))Q^*(t-1)\mathrm{diag}[O_t(b,\cdot)]Q^*(t)O_{t+1}^TV\\
&= \left(Q^*(t)O_{t+1}^TV\right)^{-1}\mathrm{diag}[O_t(b,\cdot)]\left(Q^*(t)O_{t+1}^TV\right),
\end{align*}

so that there exists an invertible $K\times K$ matrix $R := \left(Q^*(t)O_{t+1}^TV\right)^{-1}$ such that for all $b\in\{1,\dots,N\}$, 
$$\mathrm{diag}[O_t(b,\cdot)] = R^{-1}B(b)R.$$

Besides, as $O_t$ has rank $K$, there exists $(\alpha_1,\dots,\alpha_N)\in\mathbb{R}^N$ such that the eigenvalues of $B:=\sum_{b=1}^N\alpha_bB(b)$ are distinct. Hence the eigenvalue decomposition of $B$ is unique up to permutation and scaling. As $R$ diagonalizes $B$, we obtain $R$ (up to permutation and scaling) by diagonalizing $B$. Then we can deduce $O_t(b,\cdot)$ for all $b\in\{1,\dots,N\}$, up to a common permutation that we denote by $\sigma_t$. It follows that for all $t\geq 2$, the matrix $O_t$ is computable from $M_t$, $N(t)$ and $P(t)$, up to permutation of its columns, corresponding to the states. Then, since $O_t$ has full rank, we obtain $\pi^*(t)$ from $O_t$ and $L(t)$ thanks to equation \eqref{eq2}. Again, $\pi^*(t)$ is only determined up to permutation of its entries: if $R$ is replaced by $RP_{\sigma_t}$ where $P_{\sigma_t}$ is the matrix of the permutation $\sigma_t$, we get $P_{\sigma_t}^T\pi^*(t)$ instead of $\pi^*(t)$. 
We finally obtain the transition matrix:
\begin{align*}
\left(\tilde{U}^TO_t\mathrm{diag}(\pi^*(t))\right)^{-1}\tilde{U}^TN(t)V\left(O_{t+1}^TV\right)^{-1}=Q^*(t),
\end{align*}
where $\tilde{U}$ is the matrix whose columns are the left singular vectors of $N(t)$.
Replacing $O_t$ by $O_tP_{\sigma_t}$, $\pi^*(t)$ by $P_{\sigma_t}^T\pi^*(t)$ and $O_{t+1}$ by $O_{t+1}P_{\sigma_{t+1}}$ in the last equation, we obtain $P_{\sigma_t}^TQ^*(t)P_{\sigma_{t+1}}$ instead of $Q^*(t)$, which means that $Q^*(t)$ is only determined up to permutations of its lines and columns, those permutations being possibly different. 
\paragraph{} Therefore, we have proved that if the law of $(Y_t)_{t\geq 1}$ is the same under both sets of parameters $(\pi^*,Q^*(t),\nu^*_{k,t})_{1\leq t\leq T,k\in\mathsf{X}}$ and $(\tilde{\pi},\tilde{Q}(t),\tilde{\nu}_{k,t})_{1\leq t\leq T,k\in\mathsf{X}}$ with the first one satisfying Assumptions \ref{hyp-inv}-\ref{hyp-ind}, then there exists $\sigma_1^{(N)},\dots,\sigma_{T}^{(N)},\sigma_{T+1}^{(N)}$, permutations of $\{1,\dots,K\}$ such that for all $k\in\mathsf{X}$, $\tilde{\pi}_k = \pi^*_{\sigma_1^{(N)}(k)}$ and for all $t\in\{1,\dots,T\}$,
$$\tilde{Q}(t)_{kl} = Q^*(t)_{\sigma_t^{(N)}(k),\sigma_{t+1}^{(N)}(l)},\quad \tilde{O}_t^{(N)}\left(\cdot,k\right) = O_t^{(N)}\left(\cdot,\sigma_t^{(N)}(k)\right),\quad k,l\in\mathsf{X}$$
where $\tilde{O}_t$ is the analog of $O_t$ with respect to the distributions $\tilde{\nu}_{k,t}$. The relationship between $\tilde{Q}(t)$ and $Q^*(t)$ shows that for all $N\geq N_0$ and for all $t\in\{1,\dots,T\}$, $\sigma^{(N)}_t = \sigma^{(N_0)}_t$. Hence we denote by $\sigma_t$ this permutation. Thus, for all $N\geq N_0$, $\tilde{O}_t^{(N)} = O_t^{(N)}P_{\sigma_t}$. This implies that for all $t\in\{1,\dots,T\}$ and for all $k\in\mathsf{X}$, $\tilde{\nu}_{k,t} = \nu^*_{\sigma_t(k),t}$ and the theorem is proved.

\end{proof}

\paragraph{Remarks}
\begin{itemize}
\item In Theorem \ref{theo-id}, we need not assume that the second set of parameters $(\tilde{\pi},\tilde{Q}(t),\tilde{\nu}_{k,t})_{1\leq t\leq T,k\in\mathsf{X}}$ satisfies Assumptions \ref{hyp-inv}-\ref{hyp-ind} because it has to be the case. Indeed, as the two sets of parameters induce the same distribution of $(Y_1,\dots,Y_T)$, we have, using the notations of the above proof, 
$$N(t) = O_t\mathrm{diag}(\pi(t))Q(t)O_{t+1}^T = \tilde{O}_t\mathrm{diag}(\tilde{\pi}(t))\tilde{Q}(t)\tilde{O}_{t+1}^T$$
Therefore, if the second set of parameters does not satisfy Assumptions \ref{hyp-inv}-\ref{hyp-ind}, there exists some $t$ such that $\tilde{O}_t\mathrm{diag}(\tilde{\pi}(t))\tilde{Q}(t)\tilde{O}_{t+1}^T$ has not full rank. This is a contradiction since $O_t\mathrm{diag}(\pi(t))Q(t)O_{t+1}^T$ has full rank.
\item We proved that we could identify the emission distributions and the transition matrices up to permutations, these permutations depending on the time step. However it is not possible to prove that there exist a single permutation of the states that is common to all the times steps. Indeed, if we choose another labelling at time $t$, the matrix $O_t$ is replaced by $O_tP_\sigma$ and the matrix $Q(t)$ is replaced by $Q(t)P_\sigma$, with $P_\sigma$ a permutation matrix. Then we see, using \eqref{eq4}, that the matrices $M(\cdot,b,\cdot)$ remain unchanged, which means that the permutation $\sigma$ cannot be identified from the distribution of $(Y_{t-1},Y_t,Y_{t+1})$.
\item The spectral method presented above provides a non-parametric moment estimator for the parameters of a HMM. The properties of such an estimator are studied in \cite{DCGLC2017}.
\item The identifiability of the parameters $\beta$ and $\delta$ depends on the parametric form chosen for $t\mapsto Q(t)$ and $t\mapsto\theta^Y(t)$. Hence this should be studied separately for each particular version of the SHMM.

\end{itemize}

\subsection{Consistency}

In this paragraph, we prove our main result, that is the strong consistency of the maximum likelihood estimator for SHMM. Assume that we have observed $(Y_1,\dots,Y_n)$ (recall that $X_1,\dots,X_n$ are not observed). For a probability distribution $\pi$ on $\mathsf{X}$ and $\theta\in\Theta$, let $L_{n,\pi}\left[\theta;(Y_1,\dots,Y_n)\right]$ be the likelihood function when the parameter is $\theta$ and the distribution of $X_1$ is $\pi$. We define the maximum likelihood estimator by

$$\hat{\theta}_{n,\pi}:=\argmax_{\theta\in\Theta}L_{n,\pi}\left[\theta;(Y_1,\dots,Y_n)\right].$$

We will need the following assumptions:

\begin{hyp}\label{hyp-ident}
The parameter $\beta$ can be identified from the transition matrices $Q(1),\dots, Q(T)$ and the parameter $\delta$ can be identified from the emission distributions $\nu_{k,t}$.
\end{hyp}

\begin{hyp}\label{hyp-alpha}
$$\alpha:=\inf_{\theta\in\Theta}\inf_{1\leq t\leq T}\inf_{i,j\in\mathsf{X}}Q_{ij}(t)>0$$
\end{hyp}

\begin{hyp}\label{hyp-continue}
The transition probabilities (resp. the emission densities) are continuous functions of $\beta$ (resp. $\delta$).
\end{hyp}

\begin{hyp}\label{hyp-delta}
For all $y\in\mathsf{Y}$, $k\in\mathsf{X}$ and $t\in\{1,\dots,T\}$,
$$\inf_{\theta\in\Theta}f_{k,t}^\theta(y)>0,\quad \sup_{\theta\in\Theta}f_{k,t}^\theta(y)<\infty$$
\end{hyp}

\begin{hyp}\label{hyp-ct}
For $t\in\{1,\dots,T\}$ and $y\in\mathsf{Y}$, we define
$$c_t(y):=\inf_{\theta\in\Theta}\sum_{k\in\mathsf{X}}f_{k,t}^{\theta}(y),\quad d_t(y):=\sup_{\theta\in\Theta}\sum_{k\in\mathsf{X}}f_{k,t}^{\theta}(y),$$ and we assume that
$$\mathbb{E}_{\nu^{\theta^*}_{k,t}}\left[-\log c_t(Y)\right]<\infty,\quad \mathbb{E}_{\nu^{\theta^*}_{k,t}}\left[\log d_t(Y)\right]<\infty$$
\end{hyp}

Let $\mathfrak{S}_K$ be the set of permutations of $\{1,\dots,K\}$. For $\sigma = (\sigma_1,\dots,\sigma_T)\in\left(\mathfrak{S}_K\right)^T$ and $\theta\in\Theta$, let us denote by $\sigma(\theta)$ the parameter obtained from $\theta$ by swapping the states according to the permutations $\sigma_1,\dots,\sigma_T$. More precisely, we compute the transition matrices and the emission distributions corresponding to the parameter $\theta$, we swap them using the permutations $\sigma_1,\dots,\sigma_T$, and using \ref{hyp-ident}, we identify the parameter corresponding to the swapped matrices and emission distributions. This parameter is denoted by $\sigma(\theta)$. It follows from Theorem \ref{theo-id} that under Assumptions \ref{hyp-inv} to \ref{hyp-ident},
$$\Theta^*:=\{\theta\in\Theta : \mathbb{P}^\theta = \mathbb{P}^{\theta^*}\} \subset \{\sigma(\theta^*):\sigma\in \left(\mathfrak{S}_K\right)^T\}.$$

Note that due to the parametric form of the transition matrices and the emission distributions, the set $\{\sigma\in\left(\mathfrak{S}_K\right)^T : \mathbb{P}^{\theta^*} = \mathbb{P}^{\sigma(\theta^*)}\}$ may actually be much smaller than $\left(\mathfrak{S}_K\right)^T$. However, it contains at least $\{\sigma\in\left(\mathfrak{S}_K\right)^T : \sigma_1=\dots=\sigma_T\}$.

\begin{theo}\label{mainthm}
Under Assumptions \ref{hyp-inv} to \ref{hyp-ct}, for any initial distribution $\pi$ and $\mathbb{P}^{\theta^*}$-a.s., there exists $\left(\sigma^{(n)}\right)_{n\in\mathbb{N}}$ a $(\mathfrak{S}_K)^T$-valued sequence such that
$$\lim\limits_{n\to\infty}\sigma^{(n)}\left(\hat{\theta}_{n,\pi}\right)=\theta^*.$$
\end{theo}

\begin{proof}
Let us consider $(U_j,W_j)_{j\geq 0}$ the HMM defined in paragraph \ref{description}. Recall that its transition matrix is given by

$$\tilde{Q}_{uv}=Q_{u_Tv_1}(T)Q_{v_1v_2}(1)\dots Q_{v_{T-1}v_T}(T-1)$$

and its emission densities are

$$g^\theta(w\mid u ) = \prod_{t=1}^Tf^{\theta^Y}_{u_t,t}(w_t).$$

Thus the law of the process $(U_j,W_j)_{j\geq 0}$ is entirely determined by $\theta$ and $\pi$ and it is stationary under Assumption \ref{hyp-stat}. Denoting by $\mathbb{Q}^{\pi,\theta}$ the law of the process $(W_j)_{j\geq 0}$ when the parameter is $\theta$ and the distribution of $X_1$ is $\pi$, we notice that for any $\theta_1,\theta_2\in\Theta$,
$$\mathbb{Q}^{\pi,\theta_1}=\mathbb{Q}^{\pi,\theta_2}\implies \mathbb{P}^{\pi,\theta_1}=\mathbb{P}^{\pi,\theta_2}.$$

Therefore, using Theorem \ref{theo-id} and Assumption \ref{hyp-ident}, we have

\begin{equation}\label{ident2}
\mathbb{Q}^{\theta}=\mathbb{Q}^{\theta^*}\implies\exists\ \sigma\in\left(\mathfrak{S}_K\right)^T,\, \theta=\sigma(\theta^*).
\end{equation}

We notice that for all $\theta\in\Theta$, $J\geq 0$ and initial distribution $\pi$, we have :

$$\tilde{L}_{J,\pi}[\theta;(W_0,\dots,W_J)] = L_{(J+1)T,\pi}[\theta;(Y_1,\dots,Y_{(J+1)T})],$$

where $\tilde{L}_{J,\pi}$ is the likelihood function corresponding to the model $(U_j,W_j)_{j\geq 0}$ when the distribution of $X_1$ is $\pi$. Let $\tilde{\theta}_{J,\pi}$ be a maximizer of $\tilde{L}_{J,\pi}$. If we are able to prove the strong consistency of $\tilde{\theta}_{J,\pi}$, by the same arguments, for all $s\in\{0,\dots,T-1\}$, we can prove that the estimator $\tilde{\theta}^s_{J,\pi}:=\argmax_{\theta\in\Theta}{L}_{(J+1)T+s,\pi}[\theta;Y_1,\dots,Y_{(T+1)J+s}]$ is strongly consistent. From there we easily deduce that $\hat{\theta}_{n,\pi}$ is strongly consistent. Therefore it is sufficient to prove the strong consistency of the maximum likelihood estimator for the HMM $(U_j,W_j)_{j\geq 0}$. To this end, we shall use the consistency result stated in Theorem 13.14 in \cite{douc2014}. The following properties must hold in order to apply this theorem to the HMM $(U_j,W_j)_{j\geq 0}$:

\begin{enumerate}
\item $$\tilde{\alpha}:=\inf_{\theta\in\Theta}\inf_{u,v\in\mathsf{X}^T}\tilde{Q}^{\theta}_{uv}>0$$
\item For any $u,v\in\mathsf{X}^T$ and $w\in\mathsf{Y}^T$, the functions $\theta\mapsto\tilde{Q}^{\theta}(u,v)$ and $\theta\mapsto g^\theta(w\mid u)$ are continuous.
\item For all $w\in\mathsf{Y}^T$,
$$b_-(w):=\inf_{\theta\in\Theta}\sum_{u\in\mathsf{X}^T}g^\theta(w\mid u)>0,\quad b_+(w):=\sup_{\theta\in\Theta}\sum_{u\in\mathsf{X}^T}g^\theta(w\mid u)<\infty$$
\item $$\mathbb{E}^{\theta^*}\left[|\log b_+(W_0)|\right]<\infty,\quad \mathbb{E}^{\theta^*}\left[|\log b_-(W_0)|\right]<\infty$$
\\
\end{enumerate}

The three first properties are straightforward consequences of Assumptions \ref{hyp-alpha} to \ref{hyp-delta}. Let us prove that $\mathbb{E}^{\theta^*}\left[|\log b_+(W_0)|\right]<\infty$. The proof that $\mathbb{E}^{\theta^*}\left[|\log b_-(W_0)|\right]<\infty$ follows the same lines.
We have: 
$$\mathbb{E}^{\theta^*}[ | \log b_+(W_0) |]=\sum_{u\in\mathsf{X}^T}\tilde{\pi}^{\theta^*}(u)\int g^{\theta^*}(w\mid u)|\log b_+(w)|\mu^{\otimes T}(dw),$$
where $\tilde{\pi}^{\theta^*}$ is the stationary distribution associated with $\tilde{Q}^{\theta^*}$. Hence it is enough to prove that for all $u\in\mathsf{X}^T$,
$$\int g^{\theta^*}(w\mid u)|\log b_+(w)|\mu^{\otimes T}(dw)<\infty$$
We have: 
$$\int g^{\theta^*}(w\mid u)|\log b_+(w)|\mu^{\otimes T}(dw) = \int_{b_+>1} g^{\theta^*}(w\mid u)\log b_+(w)\mu^{\otimes T}(dw) + \int_{b_+<1} g^{\theta^*}(w\mid u)(-\log b_+(w))\mu^{\otimes T}(dw)$$
Using Assumption \ref{hyp-delta}, we get $\inf_{w\in\mathsf{Y}^T}b_+(w)>0$. Therefore, as $g^{\theta^*}$ is a probability density function, the second term is finite. In order to show that the first one is finite, it is enough to find a function $C$ such that for all $w\in\{b_+>1\}$, $C(w)\geq b_+(w)$ and $\int_{b_+>1} g^{\theta^*}(w\mid u)\log C(w)\mu^{\otimes T}(dw)<\infty$. Let $C(w)=\prod_{t=1}^Td_t(w_t)$.
For all $w\in\mathcal{Y}^T$, we get:
$$C(w)=\prod_{t=1}^T\sup_{\theta\in\Theta}\sum_{u_t=1}^Kf_{u_t,t}^{\theta^Y}(w_t)\geq\sup_{\theta\in\Theta}\prod_{t=1}^T\sum_{u_t=1}^Kf_{u_t,t}^{\theta^Y}(w_t)=\sup_{\theta\in\Theta}\sum_{u\in\mathsf{X}^T}\prod_{t=1}^Tf_{u_t,t}^{\theta^Y}(w_t)=\sup_{\theta\in\Theta}\sum_{u\in\mathsf{X}^T}g^\theta(w\mid u)= b_+(w)$$
In addition:
\begin{align*}
\int_{b_+>1} g^{\theta^*}(w\mid u)\log C(w)\mu^{\otimes T}(dw)&\leq\int_{\mathsf{Y}^T} g^{\theta^*}(w\mid u)\left(\sum_{t=1}^T\log d_t(w_t)\right)\mu^{\otimes T}(dw)\\
&= \sum_{t=1}^T\int_{\mathsf{Y}^T}\left(\prod_{s=1}^Tf_{u_s,s}^{\theta^*_Y}(w_s)\right)\log d_t(w_t)\mu^{\otimes T}(dw)\\
\mathrm{(Fubini)}\rightarrow&=\sum_{t=1}^T\int_{\mathsf{Y}}f_{u_t,t}^{\theta^*_Y}(w_t)\log d_t(w_t)\mu(dw_t)\\
&=\sum_{t=1}^T\mathbb{E}_{\nu_{u_t,t}^{\theta^*}}\left[\log d_t(Y)\right]\overset{\ref{hyp-ct}}{<}+\infty.
\end{align*}

Hence $\mathbb{E}^{\theta^*}\left[|\log b_+(W_0)|\right]<\infty$. Thus we can apply Theorem 13.14 in \cite{douc2014} to get that for any initial distribution $\pi$, $\mathbb{P}^{\theta^*}$-a.s., 
$$\lim\limits_{n\to\infty}d\left(\hat{\theta}_{n,\pi},\left\lbrace \theta\in\Theta:\mathbb{Q}^{\theta} = \mathbb{Q}^{\theta^*}\right\rbrace\right)=0,$$
where $d$ is a distance on $\Theta$.
 Combining this theorem with \eqref{ident2}, we obtain the strong consistency of $\tilde{\theta}_{J,\pi}$ and then Theorem \ref{mainthm} is proved.
\end{proof}
\paragraph{Remarks}
\begin{itemize}
\item The strong consistency of the MLE does not depend on the choice of the initial distribution in the computation of the likelihood function. This is due to the forgetting property of HMM under Assumptions \ref{hyp-alpha} and \ref{hyp-delta}.
\item Here we just remind the outline of the proof of Theorem 13.14 in \cite{douc2014}. Note that in the original proof, the authors use a stronger assumption than us. They assume that $\sup_{w,u}\sup_\theta g^{\theta}(w\mid u)<\infty$. However the proof is still valid with our slightly weaker assumption. Let $p^{\theta,\pi}(W_t\mid W_{t-1},\dots W_{t-m})$ the conditional probability density function of $W_t$ given $(W_{t-1},\dots,W_{t-m})$ when the parameter is $\theta$ and the initial distribution is $\pi$. The first step of the proof consists in showing that almost surely, for any $\theta\in\Theta$,
$$\lim\limits_{m\to\infty}p^{\theta,\pi}(W_t\mid W_{t-1},\dots W_{t-m})=:\Delta_{t,\infty}(\theta)$$
is well defined and does not depend on the initial distribution $\pi$. This result relies on the geometric \emph{forgetting} rate of the initial distribution by the HMM. This is where \ref{hyp-alpha} is crucial. Then we prove that $\left(\Delta_{t,\infty}(\theta)\right)_t$ is an ergodic process and that for any $\theta\in\Theta$ and any initial distribution $\pi$,
$$\frac{1}{J}\tilde{L}_{J,\pi}[\theta;(W_0,\dots,W_J)]\underset{J\to\infty}{\longrightarrow}\mathbb{E}^{\theta^*}\left[\Delta_{0,\infty}(\theta)\right]=:\ell(\theta).$$ In the last step of the proof, we show that the function $\ell$ is continuous, that for any $\theta\in\Theta$, $\ell(\theta)\leq\ell(\theta^*)$, and $\ell(\theta)=\ell(\theta^*)\implies\mathbb{Q}^{\theta}=\mathbb{Q}^{\theta^*}$.
\end{itemize}

\subsection{Applications}

In this section, we introduce two examples of SHMM and we show that under weak assumptions, Theorem \ref{mainthm} can be applied. We have seen in section \ref{ident} that one of the conditions to obtain identifiability in SHMM is the linear independence of the emission distributions. To obtain this, we will use the following lemma.

\begin{lem}\label{lem-ind}
\begin{enumerate}[1.]
\item Let $\lambda_1,\dots,\lambda_n$ be pairwise distinct positive numbers and let us denote by $\mathcal{E}(\lambda)$ the exponential distribution with parameter $\lambda$. Then the distributions $\mathcal{E}(\lambda_1),\dots,\mathcal{E}(\lambda_n)$ are linearly independent.
\item Let $m_1,\dots,m_n$ be real numbers and $\sigma_1^2,\dots,\sigma_n^2$ be pairwise distinct positive numbers. For $k\in\{1,\dots,n\}$, let us denote by $\mu_k$ the Gaussian distribution with mean $m_k$ and variance $\sigma_k^2$. Then the distributions $\mu_1,\dots,\mu_n$ are linearly independent.
\end{enumerate}
\begin{proof}
Without loss of generality, we can assume that $\lambda_1<\lambda_2<\dots <\lambda_n$. Let $(a_1,\dots,a_n)\in\mathbb{R}^n$ such that
$$a_1\mathcal{E}(\lambda_1)+\cdots + a_n\mathcal{E}(\lambda_n)=0.$$
For $x>0$, we apply this inequality to the Borel set $(-\infty,x]$ and we take the derivative with respect to $x$. We get that for any $x>0$,
\begin{equation}\label{eqlibre}
a_1\lambda_1 e^{-\lambda_1 x}+\cdots+a_n\lambda_ne^{-\lambda_n x}=0.
\end{equation}
This implies that for any $x>0$,
$$a_1\lambda_1+a_2\lambda_2e^{-(\lambda_2-\lambda_1)x}+\cdots +a_n\lambda_ne^{-(\lambda_n-\lambda_1)x}=0$$
Hence, letting $x$ go to infinity, we obtain that $a_1=0$, and equation \eqref{eqlibre} reduces to
$$a_2\lambda_2 e^{-\lambda_2 x}+\cdots+a_n\lambda_ne^{-\lambda_n x}=0.$$
Thus, step by step, we show that $a_1=\dots = a_n = 0$, which ends the proof of the first statement. The second one can be proved using the same arguments.
\end{proof}
\end{lem}

Now let us prove a second lemma that we will need to prove the invertibility of transition matrices.

\begin{lem}\label{lem-analytic}
For $n\geq 1$, let us denote by $\boldsymbol{\lambda}_n$ the Lebesgue measure on $\mathbb{R}^n$. Let $f:\mathbb{R}^n\longrightarrow\mathbb{R}$ be an analytic function and $Z(f):=\{x\in\mathbb{R}^n:\, f(x)=0\}$ its zero-set. If $\boldsymbol{\lambda}_n(Z(f))>0$ then $f\equiv 0$.
\end{lem}
\begin{proof}
We proceed by induction. Let $f:\mathbb{R}\longrightarrow\mathbb{R}$ a real analytic function such that $\boldsymbol{\lambda}_1(Z(f))>0$. Then $Z(f)$ is uncountable, hence it has an accumulation point. As $f$ is analytic, this implies that $f\equiv 0$. Now let $n\geq 2$ and assume that the result holds for analytic functions on $\mathbb{R}^{n-1}$. Let $f:\mathbb{R}^n\longrightarrow\mathbb{R}$ be an analytic function such that $\boldsymbol{\lambda}_n(Z(f))>0$. By Fubini's theorem, we have:
$$0<\boldsymbol{\lambda}_n(Z(f))=\int_{\mathbb{R}^n}\mathbbm{1}_{Z(f)}(x)dx=\int_{\mathbb{R}}dx_n\int_{\mathbb{R}^{n-1}}\mathbbm{1}_{Z(f)}(x_1,\dots,x_{n-1},x_n)dx_1\dots dx_{n-1}.$$
This implies that there exists $A\subset\R$ with $\boldsymbol{\lambda}_1(A)>0$ such that:
$$\forall x_n\in A,\, \int_{\mathbb{R}^{n-1}}\mathbbm{1}_{Z(f)}(x_1,\dots,x_{n-1},x_n)dx_1\dots dx_{n-1}>0,$$
that is
$$\forall x_n\in A,\, \boldsymbol{\lambda}_{n-1}\left[Z(f)\cap (\R^{n-1}\times\{x_n\})\right]>0.$$
Thus, for all $x_n\in A$, the function $f(\cdot,x_n):(x_1,\dots,x_{n-1})\mapsto f(x_1,\dots,x_{n-1},x_n)$ is analytic on $\mathbb{R}^{n-1}$ and vanishes on a set with positive Lebesgue measure. Therefore, for all $x_n\in A$, $f(\cdot,x_n)\equiv 0$. Let $y\in\R^{n-1}$. The function $f(y,\cdot):x_n\mapsto f(y,x_n)$ is real analytic and vanishes on $A$ with $\boldsymbol{\lambda}_1(A)>0$. Hence $f(y,\cdot)\equiv 0$. As this holds for any $y$, it shows that $f\equiv 0$. 
\end{proof}

\subsubsection{Mixtures of exponential distributions}\label{app1}

Let $(X_t,Y_t)_{t\geq 1}$ be a SHMM whose transitions are given by equation \eqref{eq-Q}, 
and whose emission distributions are mixtures of exponential distributions. Denoting by $\mathcal{E}$ the exponential distribution, the emission distribution in state $k$ and time $t$ is
$$\nu_{k,t}=\sum_{m=1}^Mp_{km}\mathcal{E}\left(\frac{\lambda_{km}}{1+\sigma_k(t)}\right),$$
where $(p_{k1},\dots,p_{kM})$ is a vector of probability, the $\lambda_{km}$ are positive and $\sigma_k(t)$ is a trigonometric polynomial whose constant term is zero and whose degree $d$ is known. This polynomial can be interpreted as a periodic scaling factor. Let $\delta_k$ the vector of coefficients of $\sigma_k$. Note that the emission densities are given by
\begin{equation}\label{eq-dens}
f_{k,t}(y) = \sum_{m=1}^Mp_{km}\frac{\lambda_{km}}{1+\sigma_k(t)}\exp\left(-\frac{\lambda_{km}}{1+\sigma_k(t)}y\right),\quad y>0.
\end{equation}

The vector of parameters of this model is $\theta = (\beta,\mathbf{p},\Lambda,\delta)$ where :
\begin{itemize}
\item $\beta = (\beta_{ijl})_{i,j,l}\in\mathbb{R}^{K\times(K-1)\times (2d+1)}$ is the parameter of the transition probabilities, given by:
$$Q_{ij}(t)\propto\exp\left[\beta_{ij1}+\sum_{l=1}^d\beta_{ij,2l}\cos\left(\frac{2\pi}{T}lt\right)+\beta_{ij,2l+1}\sin\left(\frac{2\pi}{T}lt\right)\right]$$
\item $\mathbf{p} = (p_{km})_{k,m}\in [0,1]^{K\times (M-1)}$ is the set of weights of the mixtures.
\item $\Lambda = (\lambda_{km})_{k,m}\in (0,+\infty)^{K\times M}$ is the set of parameters of the exponential distributions.
\item $\delta = (\delta_{kl})_{k,l}\in\mathbb{R}^{2d+1}$ are the coefficients of the trigonometric polynomials $(\sigma_k)_{k\in\mathsf{X}}$.
\end{itemize}

We shall make the following assumptions about the \emph{true} parameter $\theta^*$. They ensure that the mixtures have exactly $M$ components.

\begin{hyp}\label{hyp-lambda}
For all $k\in\mathsf{X}$, $\lambda^*_{k1}<\dots<\lambda^*_{kM}$.
\end{hyp}
\begin{hyp}\label{hyp-p}
For all $k\in\mathsf{X}$ and $m\in\{1,\dots,M\}$, $p_{km}^*>0$.
\end{hyp}

Let us first show that the transition matrices and the emission distributions are identifiable. Following the result of paragraph \ref{ident}, it suffices to show that Assumptions \ref{hyp-inv}-\ref{hyp-ind} are satisfied. 

\begin{itemize}
\item Clearly, for any time $t$, the transition matrix $Q^*(t)$ is irreducible as all its entries are positive.
\item Let us show that for almost every $\beta\in\mathbb{R}^{K\times(K-1)\times (2d+1)}$, $Q^*(t)$ is invertible for all $t\in\{1,\dots,T\}$. Recall Leibniz's formula for the determinant:
$$\det Q(t) = \sum_{\sigma\in\mathfrak{S}_K}\epsilon(\sigma)\prod_{i=1}^KQ_{\sigma(i),i}(t),$$
where $\mathfrak{S}_K$ is the set of permutations of $\{1,\dots,K\}$ and $\epsilon(\sigma)$ is the signature of the permutation $\sigma$. Looking at the definition of $Q(t)$, we see that $\det Q(t)=0$ if and only if
$$\sum_{\sigma\in \mathfrak{S}_K}\varepsilon(\sigma)\exp\left(\sum_{i=1}^{K-1}Z(t)\cdot\beta_{\sigma(i),i}\right)=0,$$
where $\cdot$ is the standard inner product in $\mathbb{R}^{2d+1}$, 
$$Z(t)= \left(1 , \cos\left(\frac{2\pi}{T}t\right) , \sin\left(\frac{2\pi}{T}t\right) , \cdots , \cos\left(\frac{2d\pi}{T}t\right) , \sin\left(\frac{2d\pi}{T}t\right)\right)\in\mathbb{R}^{2d+1},$$ and 
$$\beta_{\sigma(i),i} = \left( a_{\sigma(i),i,0} ; a_{\sigma(i),i,1} ; b_{\sigma(i),i,1} ; \cdots,a_{\sigma(i),i,d} ; b_{\sigma(i),i,d}\right)\in\mathbb{R}^{2d+1}.$$ For all $t\in\{1,\dots,T\}$, the function $$\phi_t:\beta\mapsto \sum_{\sigma\in \mathfrak{S}_K}\varepsilon(\sigma)\exp\left(\sum_{i=1}^{K-1}Z(t)\cdot\beta_{\sigma(i),i}\right)$$ is analytic. Therefore, by Lemma \ref{lem-analytic}, either $\phi_t(\beta)=0$ for all $\beta$, either the Lebesgue measure of its zero-set is zero. Let us define $\bar{\beta}$ by $\bar{\beta}_{ijl}=\mathbbm{1}_{i=j}\mathbbm{1}_{l=1}\log K$. The corresponding transition matrix for any time $t$ is
$$\frac{1}{2K-1}\begin{pmatrix}
K & 1 & \cdots & 1\\
1 & K & \cdots & 1\\
\vdots &  &  \ddots & \vdots\\
1 & \cdots & 1 & K 
\end{pmatrix},$$
which is invertible as it is a diagonally dominant matrix. Hence, for all $t\in\{1,\dots,T\}$, $\phi_t(\bar{\beta})\neq 0$. Thus, for any time $t$, the zero-set of $\phi_t$ is negligible. As there is a finite number of such functions, we get the desired result and Assumption \ref{hyp-inv} is satisfied.
\item The following lemma shows that Assumption \ref{hyp-ind} is easily satisfied.
 \begin{lem}\label{lem-ind2}
Assume that for all $k\in\mathsf{X}$, $\lambda^*_{k1}<\lambda^*_{k2}<\cdots<\lambda^*_{kM}$. For $t\in\{1,\dots,T\}$, let us define the set 
$$E_t=\left\lbrace\frac{\lambda^*_{km}}{1+\sigma_k(t)}:1\leq k\leq K,\, 1\leq m\leq M\right\rbrace.$$ 
If $E_t$ has cardinality at least $K$, Assumption \ref{hyp-ind} is generically satisfied.
\end{lem}
\begin{proof}
For $t\in\{1,\dots,T\}$, let us denote by $p(t)$ the cardinality of $E_t$. Thus $K\leq p(t)\leq KM$ and we can write $E_t=\{\tilde{\lambda}^*_1,\dots,\tilde{\lambda}^*_{p(t)}\}$. For all $k\in\{1,\dots K\}$, $\nu_{k,t}$ is a linear combination of exponential distributions whose parameters belong to $E_t$. Hence,
$$\nu_{k,t} = \sum_{j=1}^{p(t)}B_{kj}(t)\mathcal{E}(\tilde{\lambda}_j^*),$$
where the $B_{kj}(t)$ are among the $p_{km}^*$. The $\tilde{\lambda}_j^*$ being pairwise distinct, the family $\left(\mathcal{E}(\tilde{\lambda}_j^*)\right)_{1\leq j\leq p(t)}$ is linearly independent, using Lemma \ref{lem-ind}. Hence, the family $\left(\nu_{k,t}\right)_{k\in\mathsf{X}}$ is linearly independent if and only if the rank of the matrix $B(t)$ is $K$ (this requires that $p(t)\geq K$). This holds true except if the family $(p^*_{km})$ belongs to a set of roots of a polynomial. Moreover, the entries of $B(t)$ are among the $p^*_{km}$. This implies that the range of the map $t\mapsto B(t)$ is finite. Thus we obtain linear independence of the $\left(\nu_{k,t}\right)_{k\in\mathsf{X}}$ for all $t$, except if the $p^*_{km}$ belong to a finite union of roots of polynomials.

\end{proof}
\end{itemize}

Then, we can identify the parameter themselves from the emission distributions and the transitions matrices. 
\begin{itemize}
\item Let us denote by $\mathbb{V}$ the variance operator. For $t\in\{1,\dots, T\}$ and $k\in\mathsf{X}$, let 
$$\tilde{s}(t):=\frac{1+\sigma_k(t)}{1+\sigma_k(1)}=\sqrt{\frac{\mathbb{V}(Y_t\mid \{X_t=k\})}{\mathbb{V}(Y_1\mid \{X_1=k\})}}$$
If the emission distributions are known, $\tilde{s}(t)$ can be computed for any time step $t$. Let $c$ be the constant coefficient of the trigonometric polynomial $\tilde{s}$. It follows from the above that $c = \frac{1}{1+\sigma_k(1)}$. Hence we get $\sigma_k(t) = \frac{\tilde{s}(t)}{c}-1$. From this we can obtain $\delta^*_k$.
\item Using Assumptions \ref{hyp-lambda}-\ref{hyp-p} and equation \eqref{eq-dens}, we have 
$$\lim\limits_{y\to\infty}\frac{\log f_{k,1}^*(y)}{y}=-\frac{\lambda_{k1}^*}{1+\sigma_k(1)},$$
from which we find $\lambda_{k1}^*$. Then we can determine $p_{k1}^*$ using the following formula:
$$p_{k1}^* = \exp\left[\lim\limits_{y\to\infty}\left(\log f_{k,1}^*(y)+\frac{\lambda_{k1}^*y}{1+\sigma_k(1)} - \log\frac{\lambda_{k1}^*}{1+\sigma_k(1)}\right)\right].$$
Step by step, following the same method, we identify the rest of the parameters of the emission distributions.
\item In order to retrieve $\beta$ from the transition matrices $(Q(t))_{1\leq t\leq T}$, notice that for any $t\in\{1,\dots,T\}$ and for any $i\in\mathsf{X}$,
$$\sum_{j=1}^{K-1}Q_{ij}(t) = \frac{\sum_{j=1}^{K-1}\exp\left[\sum_{l=0}^d\left(a_{ijl}\cos\left(\frac{2\pi lt}{T}\right)+b_{ijl}\sin\left(\frac{2\pi lt}{T}\right)\right)\right]}{1+\sum_{j=1}^{K-1}\exp\left[\sum_{l=0}^d\left(a_{ijl}\cos\left(\frac{2\pi lt}{T}\right)+b_{ijl}\sin\left(\frac{2\pi lt}{T}\right)\right)\right]},$$
 so that
$$\sum_{j=1}^{K-1}\exp\left[\sum_{l=0}^d\left(a_{ijl}\cos\left(\frac{2\pi lt}{T}\right)+b_{ijl}\sin\left(\frac{2\pi lt}{T}\right)\right)\right]=\frac{\sum_{j=1}^{K-1}Q_{ij}(t)}{1-\sum_{j=1}^{K-1}Q_{ij}(t)}.$$
Then, for any $j\in\{1,\dots,K-1\}$, 
$$\exp\left[\sum_{l=0}^d\left(a_{ijl}\cos\left(\frac{2\pi lt}{T}\right)+b_{ijl}\sin\left(\frac{2\pi lt}{T}\right)\right)\right] = Q_{ij}(t)\frac{\sum_{j'=1}^{K-1}Q_{ij'}(t)}{1-\sum_{j'=1}^{K-1}Q_{ij'}(t)}.$$
As a trigonometric polynomial of degree $d$ has at most $2d$ zeros over a period, this implies that we can retrieve $\beta$, as long as $T>2d$. Hence Assumption \ref{hyp-ident} is satisfied.
\end{itemize}

\paragraph{}In order to prove the strong consistency of the maximum likelihood estimator, it remains to check that Assumptions \ref{hyp-alpha} to \ref{hyp-ct} are satisfied. 
\begin{itemize}
\item Assuming that there exists $\beta_{\min}$ and $\beta_{\max}$ such that for all $\theta\in\Theta$ and for all $i,j,l$, $\beta_{ijl}\in [\beta_{\min},\beta_{\max}]$, Assumption \ref{hyp-alpha} is satisfied.
\item Assumption \ref{hyp-continue} is clearly satisfied. 
\item Assumption \ref{hyp-delta} is satisfied provided that:
\begin{itemize}
\item $\inf_{\theta\in\Theta}\inf_{k\in\mathsf{X},m\in\{1,\dots,M\}}p_{km}>0$
\item there exists positive numbers $\lambda_{\min}$ and $\lambda_{\max}$ such that for all $\theta\in\Theta$ and for all $k,m$, $\lambda_{km}\in~ [\lambda_{\min},\lambda_{\max}]$
\item there exists positive numbers $\sigma_{\min}$ and $\sigma_{\max}$ such that for all $\theta\in\Theta$, for all $k\in\mathsf{X}$ and for all $t\in\{1,\dots,T\}$, $\sigma_k(t)\in~ [\sigma_{\min},\sigma_{\max}]$. Implicitly, this is a boundedness condition on the parameter $\delta$.
\end{itemize}
\item Under the same boundedness assumptions, we obtain Assumption \ref{hyp-ct}.
\end{itemize}

Thus, under weak conditions on the parameters, we can apply our identifiability and convergence results to this particular model.

\subsubsection{Mixtures of Gaussian distributions}\label{app2}

We choose the same transition matrices as in the previous example, and the emission distribution in state $k$ at time $t$ writes:
$$\nu_{k,t}=\sum_{m=1}^Mp_{km}\mathcal{N}\left(m_k(t),\sigma_{km}^2\right)$$
where $m_k$ is a trigonometric polynomial with (known) degree $d$, $(p_{k1},\dots,p_{kM})$ is a vector of probability and $\mathcal{N}(m,\sigma^2)$ refers to the Gaussian distribution with mean $m$ and variance $\sigma^2$. The following lemma ensures that Assumption \ref{hyp-ind} is easily satisfied.

\begin{lem}\label{lem-ind3}
Assume that for all $k\in\mathsf{X}$, $\sigma^2_{k1},\dots,\sigma^2_{kM}$ are pairwise distinct. Let 
$$E = \{\sigma^2_{km}, 1\leq m\leq M,\, 1\leq k\leq K\}.$$
If $E$ has at least $K$ elements, then, for all $t\in\{1,\dots, T\}$, the distributions $\nu_{1,t},\dots,\nu_{K,t}$ are linearly independent for almost every choice of $\mathbf{p}$.
\end{lem}
\begin{proof}
Using the second statement of Lemma \ref{lem-ind}, the proof is the same as in Lemma \ref{lem-ind2}.
\end{proof}

Hence we can identify the transition matrices and the emission distributions, up to state labelling. Thus, using the fact that finite Gaussian mixtures are identifiable, we can identify, for each state $k$ and each time step $t$, the vector $(p_{k1},\dots,p_{kM})$, the mean $m_k(t)$ and the variances $\sigma_{km}^2$, up to permutation of the components of the mixture. Finally, for each $k$, we can identify the coefficients of the trigonometric polynomial $m_k(\cdot)$ from its values $(m_k(1),\dots,m_k(T))$, so that Assumption \ref{hyp-ident} is satisfied.
Then, under boundedness conditions on the parameters that are very similar to those of the previous example, we see that Assumptions \ref{hyp-alpha} to \ref{hyp-ct} are satisfied. Hence the strong consistency of the maximum likelihood estimator.

\section{Simulation study}\label{sec3}

\subsection{Computation of the maximum likelihood estimator}\label{em}

We have already shown the strong consistency of the maximum likelihood estimator. This paragraph deals with its practical computation. Assume that we have observed a trajectory of the process $(Y_t)_{t\geq 1}$ with length $n$. Let $X:=(X_1,\dots,X_n)$ and $Y:=(Y_1,...,Y_n)$ and recall that $X$ is not observed. The likelihood function with initial distribution $\pi$ is then
$$L_{n,\pi}[\theta;Y]=\sum_{\mathbf{x}}\pi_{x_1}f^{\theta_Y}_{x_1,1}(Y_1)\prod_{t=2}^nQ_{x_{t-1}x_t}(t-1)f^{\theta_Y}_{x_t,t}(Y_t),$$
where $\mathbf{x}=\left(x_1,\dots,x_n\right)$. As $X$ is not observed, we use the Expectation Maximization (EM) algorithm to find a local maximum of the log-likelihood function. The EM algorithm is a classical algorithm to perform maximum likelihood inference with incomplete data. See \cite{dempster77} for a general formulation of the EM algorithm and \citep{baum1970} for its application to HMM. For any initial distribution $\pi$, we define the \emph{complete} log-likelihood by:

\begin{align*}
\log L_{n,\pi}\left[\theta;(X,Y)\right]:=\log \pi_{X_1}+\sum_{t=1}^{n-1}\log Q_{X_tX_{t+1}}(t)+\sum_{t=1}^n\log f^{\theta_Y}_{X_t,t}\left(Y_t\right).
\end{align*}

This would be the log-likelihood function if $X$ were observed. The algorithm starts from an initial vector of parameters $(\theta^{(0)},\pi^{(0)})$ and alternates between two steps to construct a sequence of parameters $\left(\theta^{(q)},\pi^{(q)}\right)_{q\geq 0}.$ 
\begin{itemize}
\item The \textbf{E} step is the computation of the \emph{intermediate quantity} defined by:
$$\mathbf{Q}\left[\left(\theta,\pi\right),\left(\theta^{(q)},\pi^{(q)}\right)\right]:=\mathbb{E}^{\pi^{(q)},\theta^{(q)}}\left[\log L_{n,\pi}\left(\theta;(X,Y)\right)\mid Y\right].$$
This requires to compute the \emph{smoothing probabilties}, that are the \emph{a posteriori} distributions of $X$ given $Y$. More precisely, the following quantities are to be computed:
$$\pi_{t\mid n}^{(q)}(k):=\mathbb{P}^{\pi^{(q)},\theta^{(q)}}\left(X_t=k\mid Y\right)$$
 for all $k\in\mathsf{X}$ and $1\leq t\leq n$, and
$$\pi_{t,t+1 \mid n}^{(q)}(k,l) := \mathbb{P}^{\pi^{(q)},\theta^{(q)}}\left(X_t = k,X_{t+1} = l\mid Y\right)$$ for $k,l\in\mathsf{X}$ and $1\leq t\leq n-1$. The computation of the smoothing probabilities can be done efficiently using the \emph{forward-backward} algorithm. See \cite{rabiner86} or \cite{cappe2009} for a description of this algorithm in the framework of HMM. The adaptation of the forward-backward algorithm for SHMM is straightforward. The intermediate quantity writes:

\begin{equation*}
\begin{split}
\mathbf{Q}\left[\left(\theta,\pi\right),\left(\theta^{(q)},\pi^{(q)}\right)\right]&=\mathbb{E}^{\pi^{(q)},\theta^{(q)}}\left[\log L_{n,\pi}\left(\theta;(X,Y)\right)\mid Y\right]\\
 &= \sum_{k=1}^K\pi_{1\mid n}^{(q)}(k)\log\pi_k\\
 &+ \sum_{t=1}^{n-1}\sum_{k=1}^K\sum_{l=1}^K\pi_{t,t+1\mid n}^{(q)}(k,l)\log Q_{kl}(t)\\
 &+ \sum_{t=1}^n\sum_{k=1}^K\pi_{t\mid n}^{(q)}(k)\log f^{\theta_Y}_{k,t}(Y_t).
\end{split}
\end{equation*}

\item The \textbf{M} step consists in finding $\left(\theta^{(q+1)},\pi^{(q+1)}\right)$ maximizing the function $\left(\theta,\pi\right)\mapsto \mathbf{Q}\left[\left(\theta,\pi\right),\left(\theta^{(q)},\pi^{(q)}\right)\right]$, or at least increasing it. Depending on the specific models chosen for $t\mapsto Q(t)$ and $t\mapsto\theta^Y(t)$ it is sometimes possible to find an analytic formula for the solution of this maximization problem. However, in most cases, a numerical optimization algorithm is required.
\end{itemize}
It can be shown that the sequence of likelihoods $L_n\left[\theta^{(q)};Y\right]$ is increasing and that under regularity conditions, it converges to a local maximum of the likelihood function \citep{wu83}. We alternate the two steps of the EM algorithm until we reach a stopping criterion. For example, we can stop the algorithm when the relative difference $\frac{L_{n,\pi^{(q)}}(\theta^{(q+1)};Y)-L_{n,\pi^{(q)}}(\theta^{(q)};Y)}{L_{n,\pi^{(q)}}(\theta^{(q)};Y)}$ drops below some threshold $\varepsilon$. The last computed term of the sequence $\left(\theta^{(q)}\right)_{q\geq 0}$ is then an approximation of the maximum likelihood estimator. However, if the EM algorithm does converge, it only guarantees that the limit is a \emph{local} maximum of the likelihood function, which may not be global. Therefore it is a common practice to run the algorithm a large number of times, starting from different (e.g. randomly chosen) initial points and select the parameter with the largest likelihood. In \cite{biernacki2003}, the authors compare several procedures to initialize the EM algorithm, using variants such as SEM \citep{broniatowski1983}. Introducing randomness in the EM algorithm provides a way to escape from local maxima.

\subsection{Example}

\paragraph{Model}
Let us consider the following SHMM $(X_t,Y_t)_{t\geq 1}$:

\begin{itemize}
\item Two states: $K=2$
\item $T = 365$
\item $Q(t)$, the transition matrix at time $t$ is determined by its first column, given by:
$$Q_{i1}(t)=\frac{\exp\left(\beta_{i1}+\beta_{i2}\cos\left(\frac{2\pi t}{T}\right)+\beta_{i3}\sin\left(\frac{2\pi t}{T}\right)\right)}{1+\exp\left(\beta_{i1}+\beta_{i2}\cos\left(\frac{2\pi t}{T}\right)+\beta_{i3}\sin\left(\frac{2\pi t}{T}\right)\right)}$$
Thus the law of the Markov chain $(X_t)_{t\geq 1}$ is determined by $\pi$ (the distribution of $X_1$) and the $\beta_{il}$ for $1\leq i\leq K$ and $1\leq l\leq 3$.
\item The emission distributions are Gaussian:
$$Y_t\mid\{X_t=k\}\sim\mathcal{N}\left(m_k(t),\sigma^2_k\right)$$
where the mean $m_k(t)$ is given by
$$m_k(t) = \mu_k + \delta_{k1}\cos\left(\frac{2\pi t}{T}\right) + \delta_{k2}\sin\left(\frac{2\pi t}{T}\right)$$
Thus the parameters of the emission distributions are the $\mu_k$, $\delta_{k1}$, $\delta_{k2}$ and $\sigma^2_k$, for $1\leq k\leq K$. Note that this is a special case of the model introduced in paragraph \ref{app2}, with $M=d=1$.
\end{itemize}
$\theta$ is the vector containing both the parameters of the transitions and the parameters of the emission distributions. Then for any choice of $\theta$ and $\pi$, it is easy to simulate a realization of $(X_t,Y_t)_{1\leq t\leq n_{\max}}$. First we simulate the Markov chain $(X_t)$, then we simulate $(Y_t)$ conditionally to $(X_t)$. We chose $n_{\max} = 200000$.\\

\paragraph{True parameters}

\begin{itemize}
\item The transition probabilities are given by
$$\pi^*=\begin{pmatrix}
0.5 & 0.5
\end{pmatrix},\quad 
\beta^*_1:=(\beta^*_{1l})_{1\leq l\leq 3} = \begin{pmatrix}
1 & 0.7 & 0.5
\end{pmatrix},\quad
\beta^*_2:=(\beta^*_{2l})_{1\leq l\leq 3} = \begin{pmatrix}
-1 & -0.6 & 0.7
\end{pmatrix}$$
The graphs of the functions $t\mapsto Q^*_{11}(t)$ (black) and $t\mapsto Q^*_{22}(t)$ (red) for $1\leq t\leq 365$ are depicted in Figure \ref{qstar}.

\begin{figure}[H]
\centering
\includegraphics[scale=0.5]{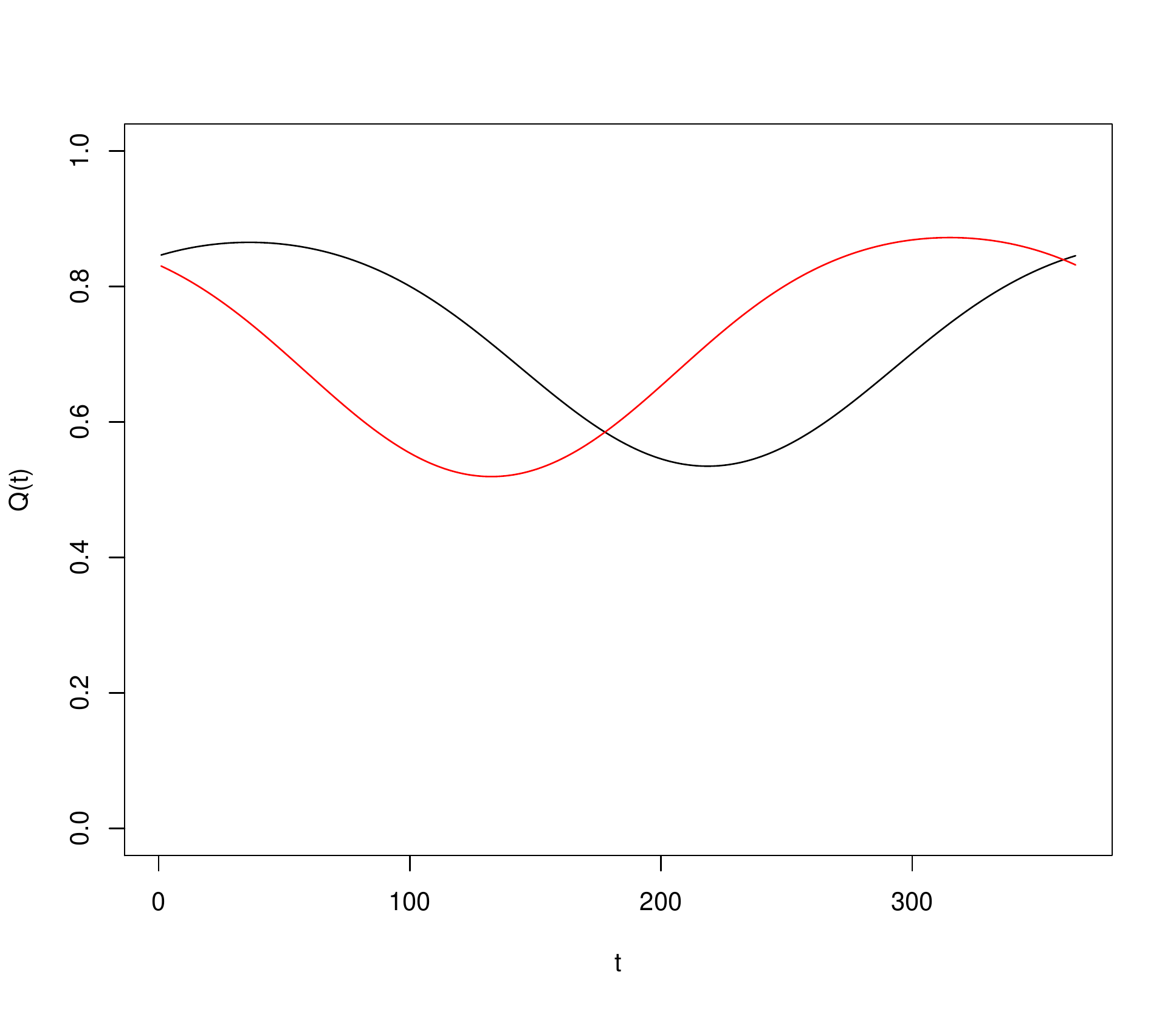}
\caption{Transition probabilities $Q^*_{11}(t)$ (black) and $Q^*_{22}(t)$ (red)}\label{qstar}
\end{figure}

\item The parameters of the emission distributions are given by
$$\mu^* = \begin{pmatrix}
-1 & 2
\end{pmatrix},\quad (\sigma^2)^* = \begin{pmatrix}
1 & 0.25
\end{pmatrix},\quad \delta_1^* = \begin{pmatrix}
2.5 & 4
\end{pmatrix},\quad \delta_2^* = \begin{pmatrix}
-1.5 & 3.5
\end{pmatrix}$$

\end{itemize}

Figure \ref{ysim} depicts a simulation of $(Y_t)_{1\leq t\leq 1000}$ using the parameter $\theta^*$. The lines correspond to the conditionnal means $m_k(t)$.

\begin{figure}[H]
\centering
\includegraphics[scale=0.4]{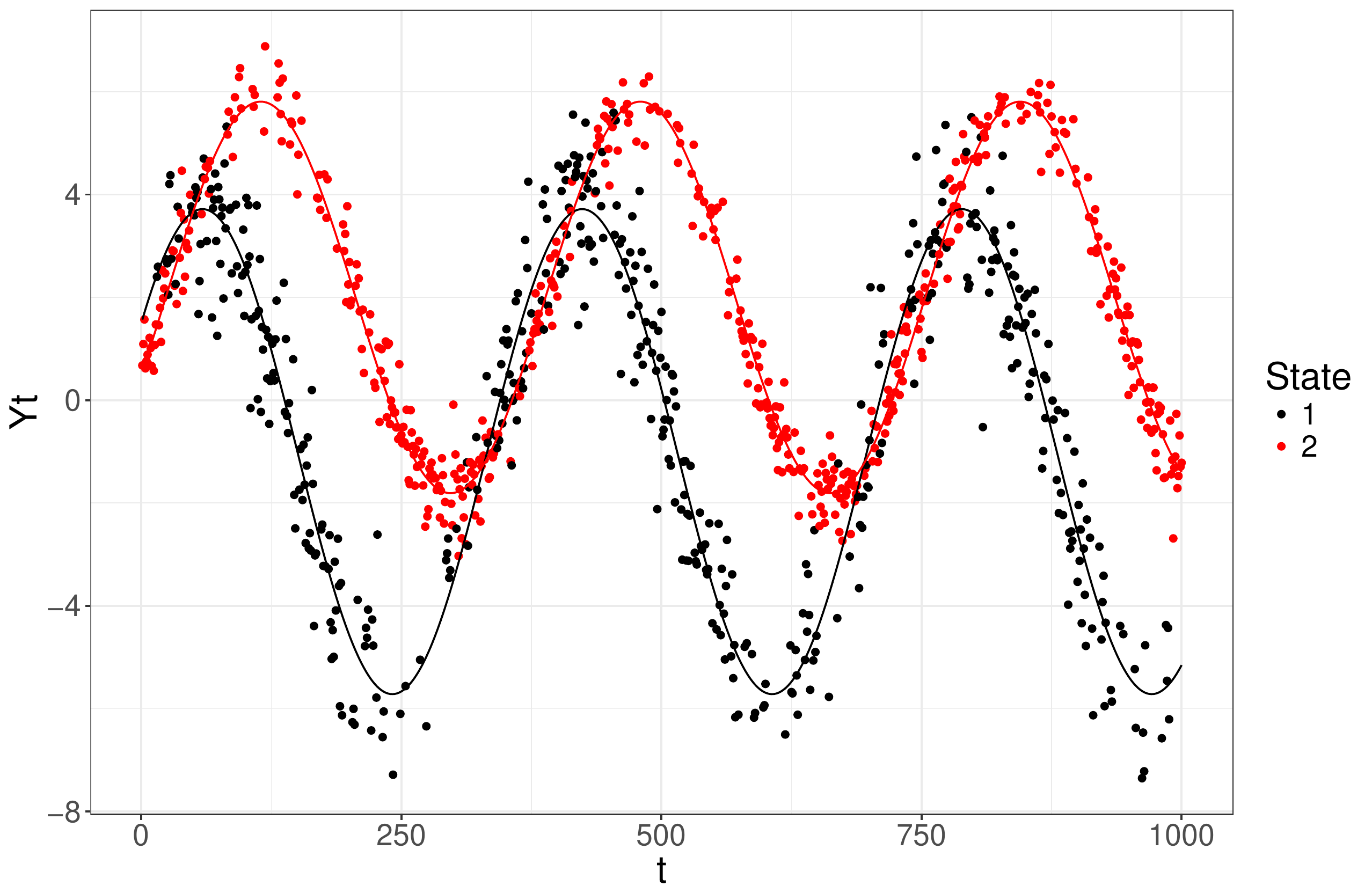}
\caption{Realization of $Y_1,\dots,Y_{1000}$}\label{ysim}
\end{figure}

\paragraph{Estimation}

In order to compute the MLE in this model, we use the EM algorithm described in paragraph \ref{em}. The estimation procedure is described below:
\begin{enumerate}
\item Select a random initial parameter and run the EM algorithm using $Y_1,\dots,Y_{500}$, with $50$ iterations starting from this initial point.
\item Repeat the first step $30$ times.
\item Among the $30$ initial points candidates, select the one that led to the largest log-likelihood after the EM.
\item Run a long EM using the selected initial point and $Y_1,\dots,Y_n$. Stop when the relative difference in log-likelihoods drops below $10^{-7}$. Then the result of the last $\mathbf{M}$-step is $\hat{\theta}_{n}$.
\end{enumerate}

Hence we can compute a sequence of MLE $\left(\hat{\theta}_{n_p}\right)_{p\geq 1}$. We chose $n_p = 100p$, for $1\leq p\leq 150$.

\paragraph{Results}

The following graphs show the estimated parameters for the emission distributions, i.e. the $\mu_k$, $\delta_k$ and $\sigma^2_k$. The dashed lines represent the true parameters. 

\begin{figure}[H]
\centering
\includegraphics[scale=0.4]{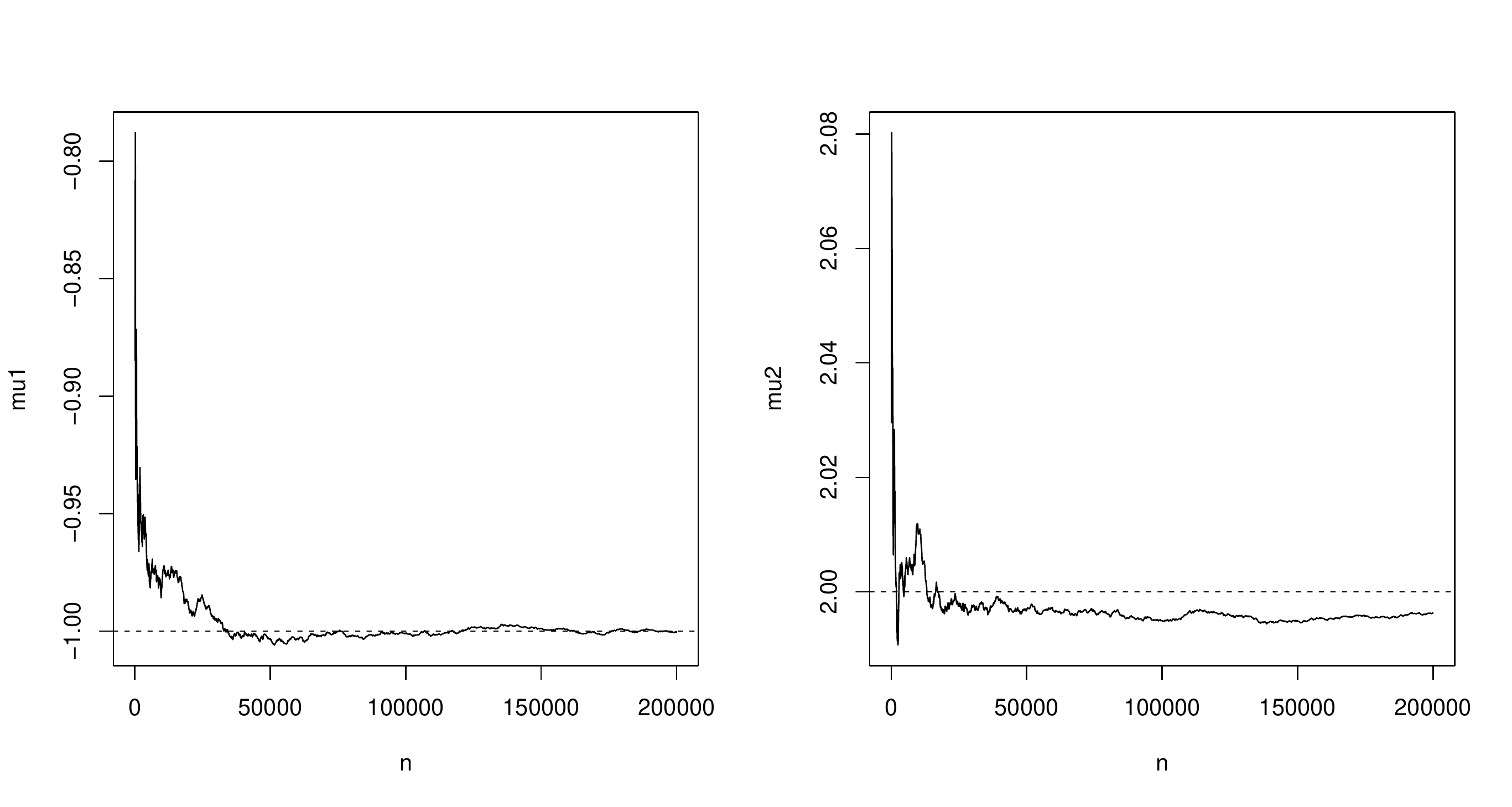}
\caption{Estimators $\hat{\mu}_{1,n}$ (left) and $\hat{\mu}_{2,n}$ (right)}\label{muhat}
\end{figure}

\begin{figure}[H]
\centering
\includegraphics[scale=0.5]{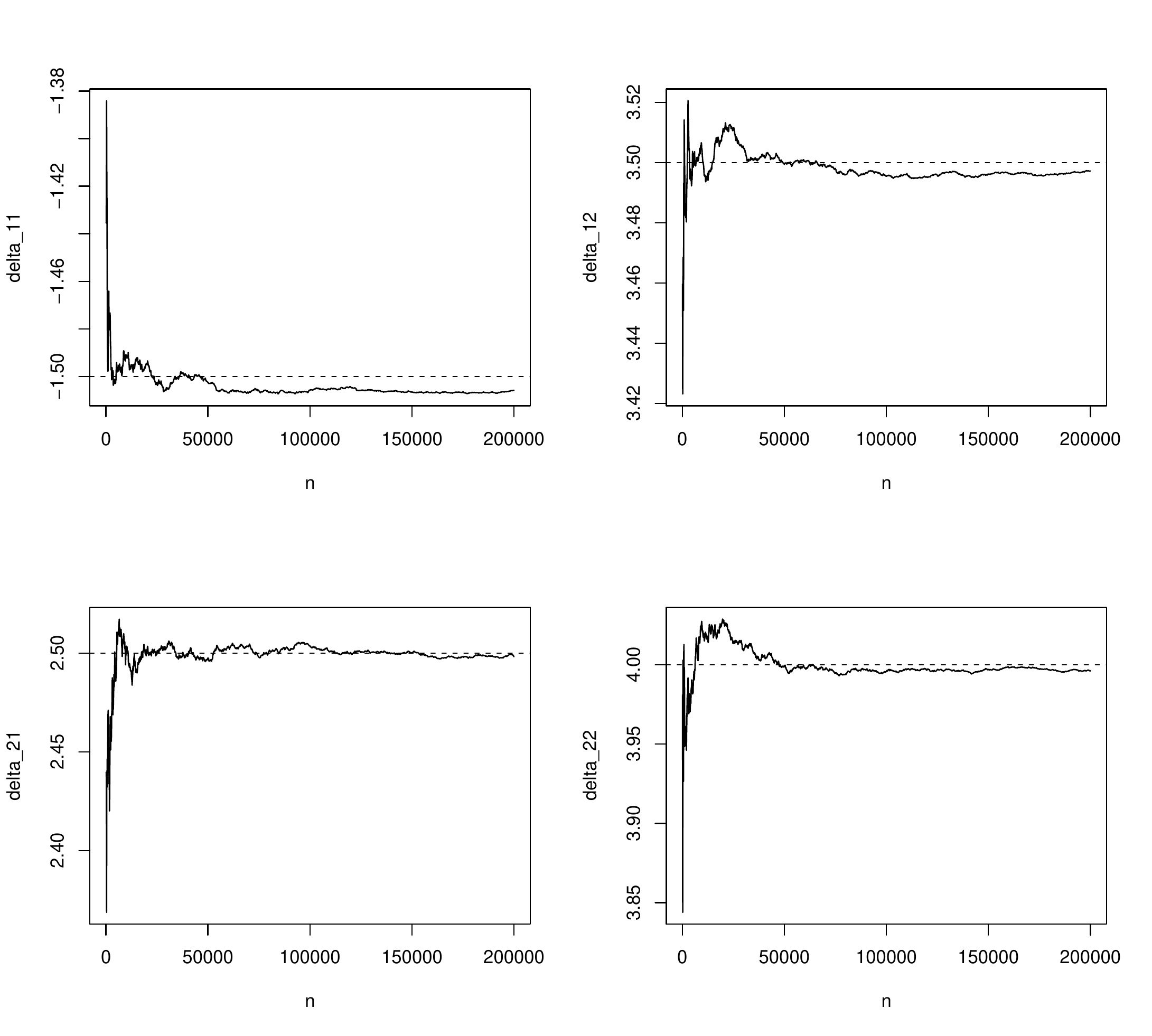}
\caption{Estimators $\hat{\delta}_{k,n}$}\label{deltahat}
\end{figure}

\begin{figure}[H]
\centering
\includegraphics[scale=0.4]{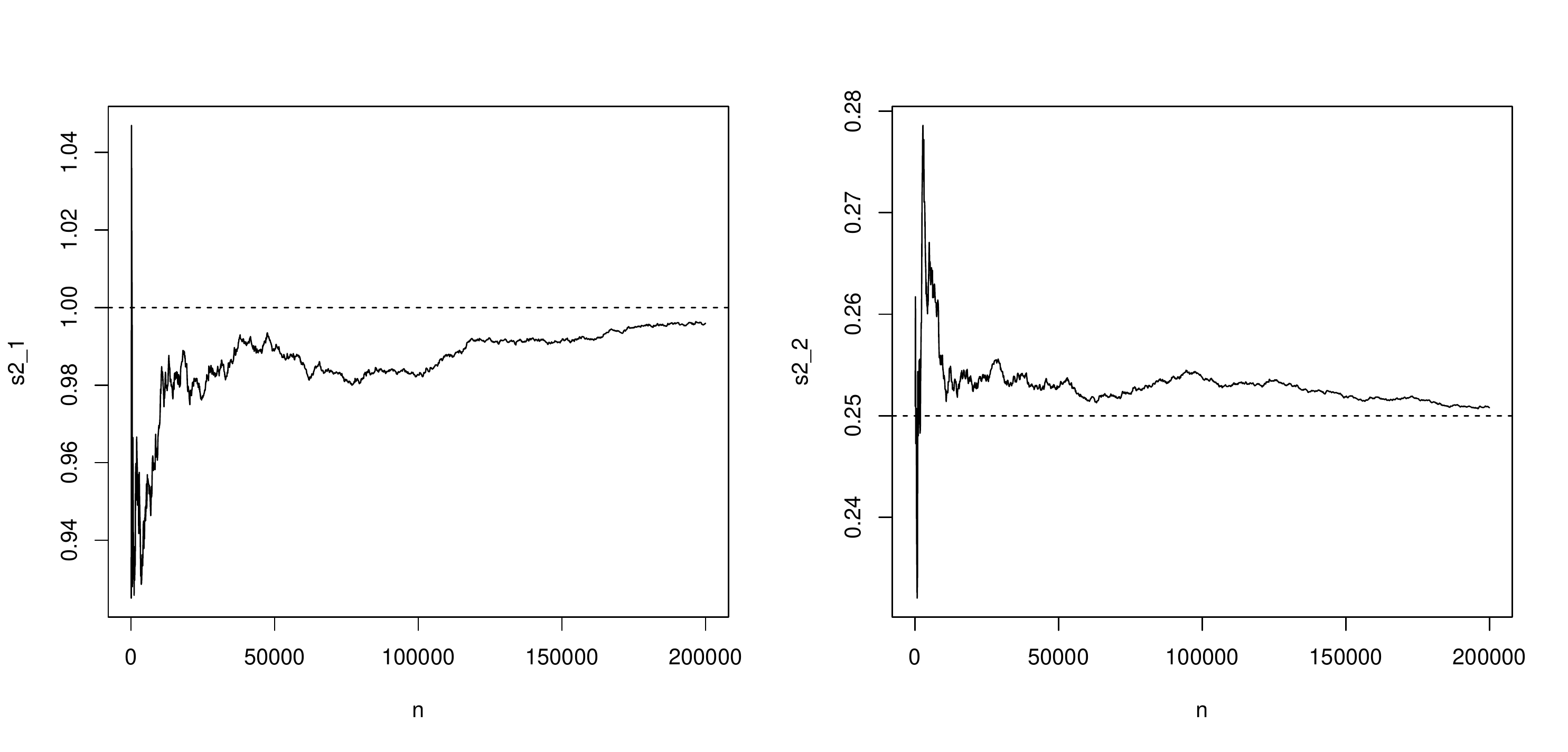}
\caption{Estimators $\hat{\sigma}^2_{1,n}$ (left) and $\hat{\sigma}^2_{2,n}$ (right)}\label{s2hat}
\end{figure}

We can also check that the periodic means $m_k(t)$ have been well estimated. In Figure \ref{mkhat} are drawn the graphs of the true means (solid line) and their estimated counterparts (dashed line), for each state.

\begin{figure}[H]
\centering
\includegraphics[scale=0.4]{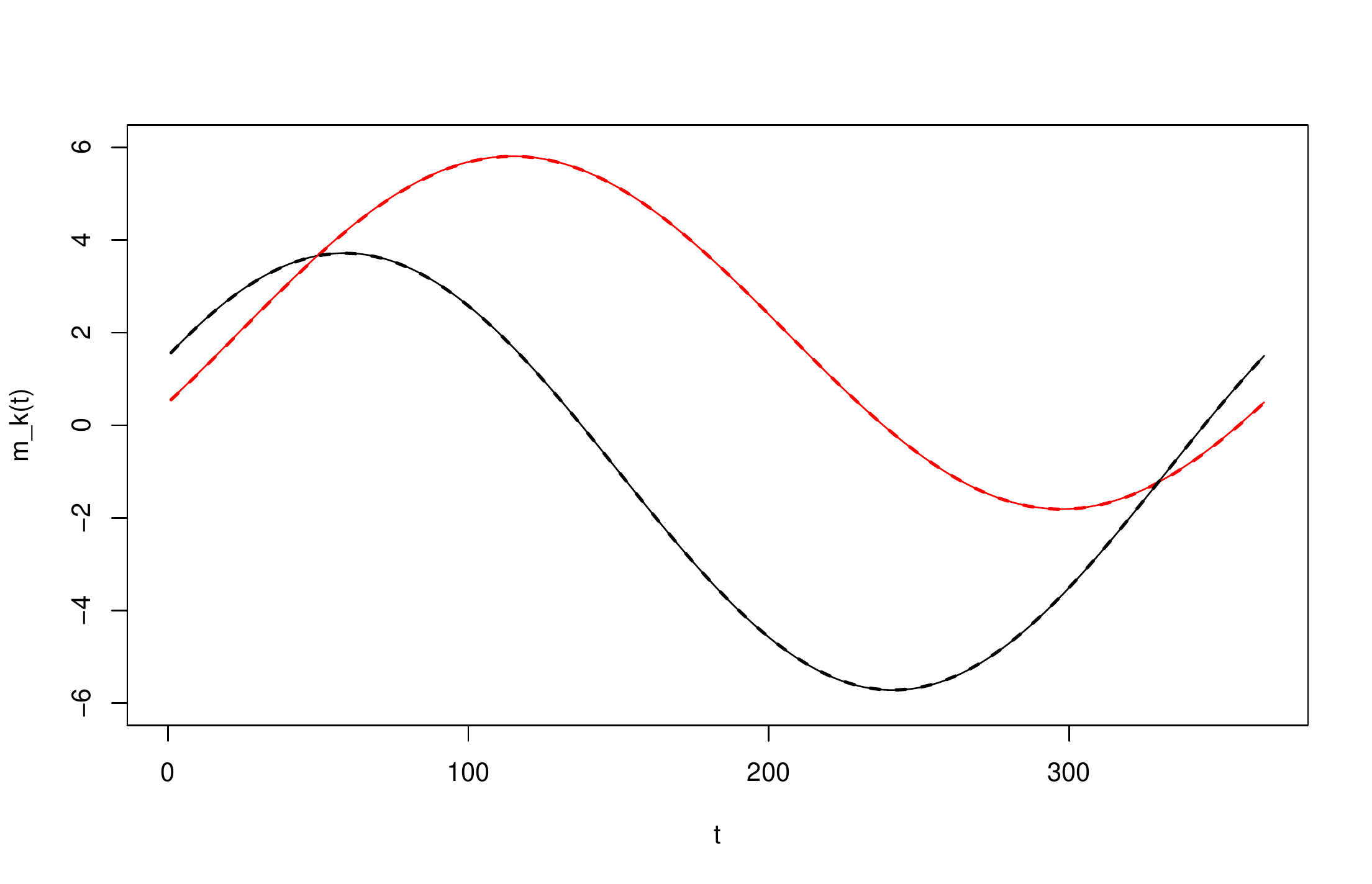}
\caption{True means $m_k(t)$ and their estimators $\hat{m}_k(t)$ (here computed with $n=n_{\max}$)}\label{mkhat}
\end{figure}

The coefficients $\beta_{kl}$ for the transition matrices, as well as the transition matrices themselves, are well estimated, as shown in Figures \ref{bethat} and \ref{Qhat}.

\begin{figure}[H]
\centering
\includegraphics[scale=0.5]{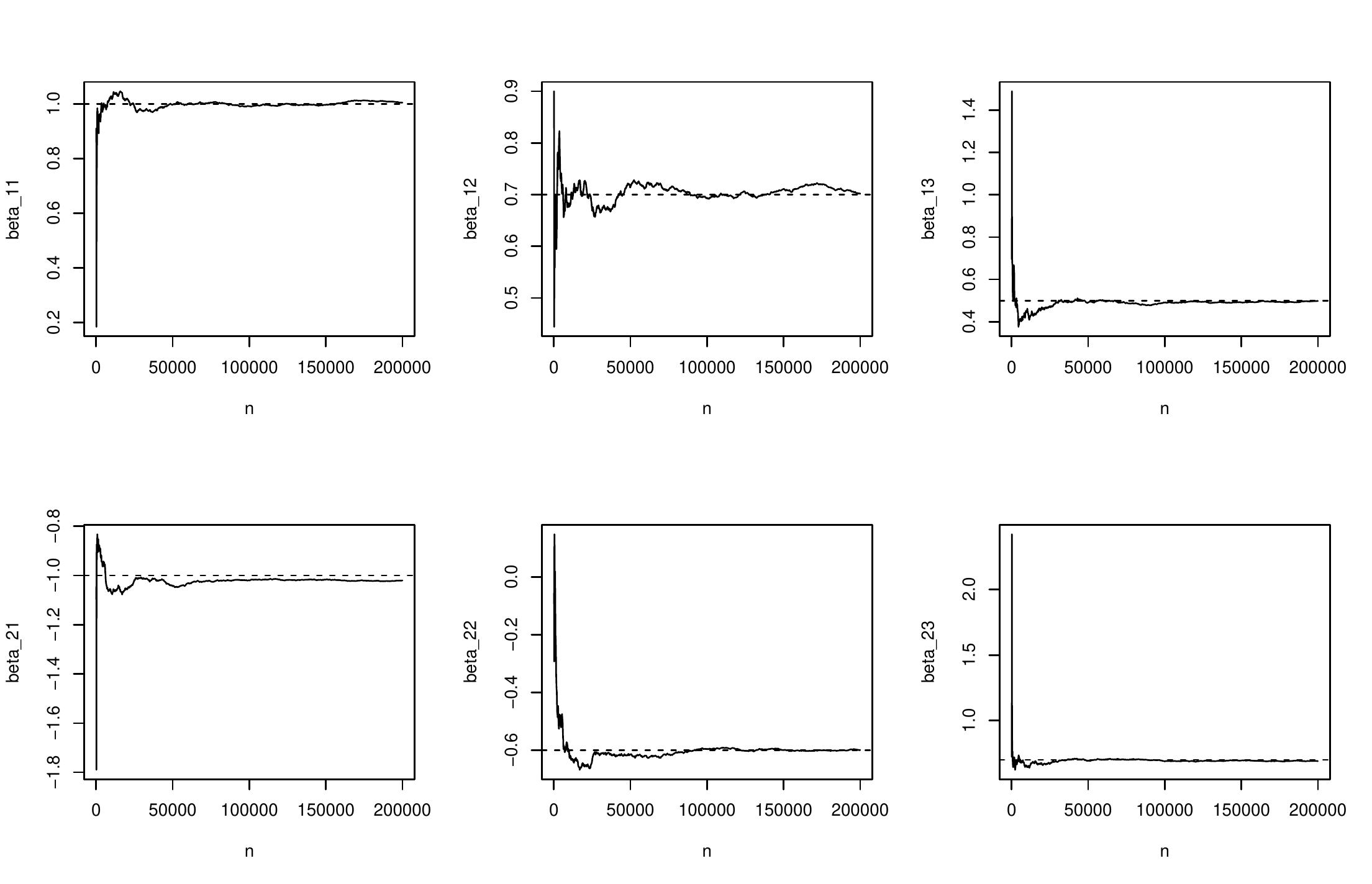}
\caption{Estimators $\hat{\beta}_{kl}$}\label{bethat}
\end{figure}

\begin{figure}[H]
\centering
\includegraphics[scale=0.4]{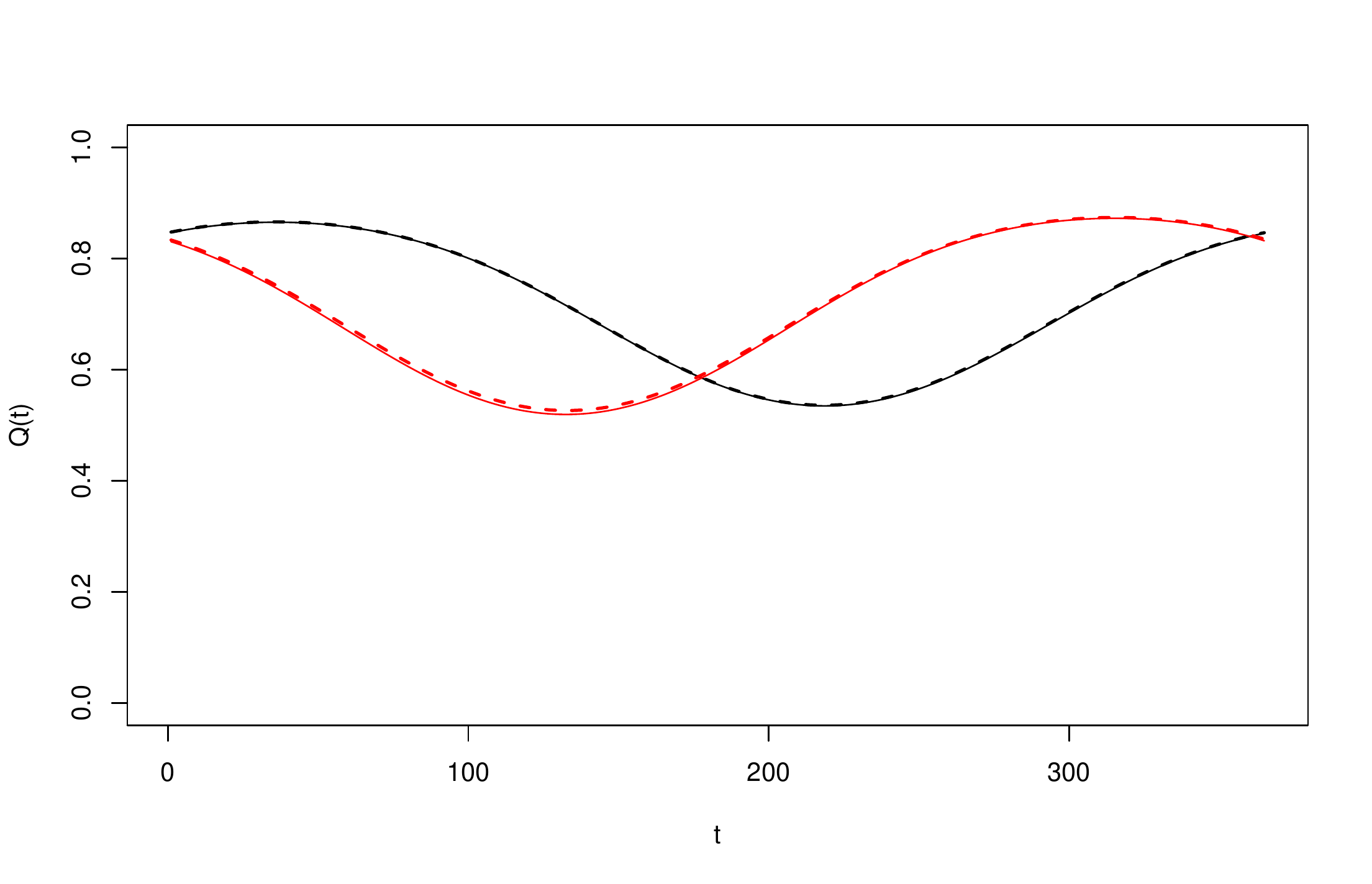}
\caption{True transitions $Q_{11}(t)$ and $Q_{22}(t)$ (solid line) and their estimated counterparts (dashed line, computed with $n=n_{\max}$)}\label{Qhat}
\end{figure}

\paragraph{Remarks}
\begin{itemize}
\item We had to swap the two states before producing these graphs. This illustrate the fact that the emission distributions and transitions matrices are identifiable only up to label swapping.
\item It is not necessary to have such a long times series ($200000$ observations) to produce good estimates. The figures above show that we can obtain good results with only $20000$ observations.
\end{itemize}

\section{Application to precipitation data}\label{apprain}

This work was motivated by the design of a \emph{stochastic weather generator}. A \emph{stochastic weather generator} \citep{katz1996} is a statistical model used whenever we need to quickly produce synthetic time series of weather variables. These series can then be used as input for physical models (e.g. electricity consumption models, hydrological models...), to study climate change, to investigate on extreme values \citep{yaoming2004}...  A good weather generator produces times series that can be considered \emph{realistic}. By realistic we mean that they mimic the behaviour of the variables they are supposed to simulate, according to various criteria. For example, a temperature generator may need to reproduce daily mean temperatures, the seasonality of the variability of the temperature, its global distribution, the distribution of the extreme values, its temporal dependence structure... and so on. The criteria that we wish to consider largely depend on applications. In this section, we introduce a univariate stochastic weather generator that focuses on rainfall.

\paragraph{Data} We use data from the \emph{European Climate Assessment and Dataset} (ECA\&D project: \url{http://www.ecad.eu}). It consists in daily rainfalls measurements (in millimeters) at the weather station of Bremen, Germany, from 1/1/1950 to 12/31/2015. We remove the 16 February 29 so that every year of the period of observation has 365 days. Thus there are 24090 data points left. Missing data are replaced by drawing at random a value among those corresponding to the same day of the year. The distribution of daily precipitation amounts naturally appears as a mixture of a mass at $0$ corresponding to dry days, and a continuous distribution with support on $\mathbb{R}_+$ corresponding to the intensity of precipitations on rainy days. Also, the data exhibits a seasonal behaviour with an annual cycle. Rainfalls tend to be less frequent and heavier in summer than in winter. Hence, a simple HMM cannot be used to model this data, but using a SHMM seems appropriate.

\paragraph{Model}

We use a model that is very similar to the one introduced in paragraph \ref{app1}. To account for dry days, for each state, we replace the first component of the mixture of exponential distributions by a Dirac mass at $0$, so that the emission distribution in state $k$ is
$$\nu_{k} = p_{k1}\delta_0 + \sum_{m=2}^Mp_{km}\mathcal{E}\left(\lambda_{km}\right).$$
Notice that the emission densities do not depend on $t$, as introducing periodic emission distributions is not necessary to generate realalistic time series of precipitations. It is enough to consider periodic transitions. Here the dominating measure is $\mu = \delta_0 + \boldsymbol{\lambda}$ where $\boldsymbol{\lambda}$ is the Lebesgue measure over $(0,+\infty)$. Hence the emission densities are given by
$$f_{k}(y) = p_{k1}\mathbbm{1}_{y=0} + \sum_{m=2}^Mp_{km}\lambda_{km}e^{-\lambda_{km}y}\mathbbm{1}_{y>0}.$$
Recall that $(p_{k1},\dots,p_{kM})$ is a vector of probability and the transition probabilities between the hidden states are given by equation \eqref{eq-Q}.

\paragraph{Results} We estimate the parameters of the model by maximum likelihood inference, using the EM algorithm described in paragraph \ref{em}.
We choose $K=4$ (four states), $M = 3$ (mixtures of two exponential distributions and a Dirac mass at $0$) and $d=2$ (complexity of the seasonal components) as these parameters give good validation results. Figure \ref{dens} displays the estimated emission densities. They correspond to the following estimators:

$$\hat{\Lambda}=\begin{pmatrix}
3.455 & 3.455 \\ 
  1.162 & 1.825 \\ 
  2.330 & 0.481 \\ 
  0.086 & 0.214 \\ 
\end{pmatrix},\quad \hat{\mathbf{p}}=\begin{pmatrix}
0.983 & 0.003 & 0.014 \\ 
  0.749 & 0.025 & 0.226 \\ 
  0.032 & 0.258 & 0.709 \\ 
  0.029 & 0.059 & 0.912 \\ 
\end{pmatrix}$$

\begin{figure}[H]
\centering
\includegraphics[scale=0.5]{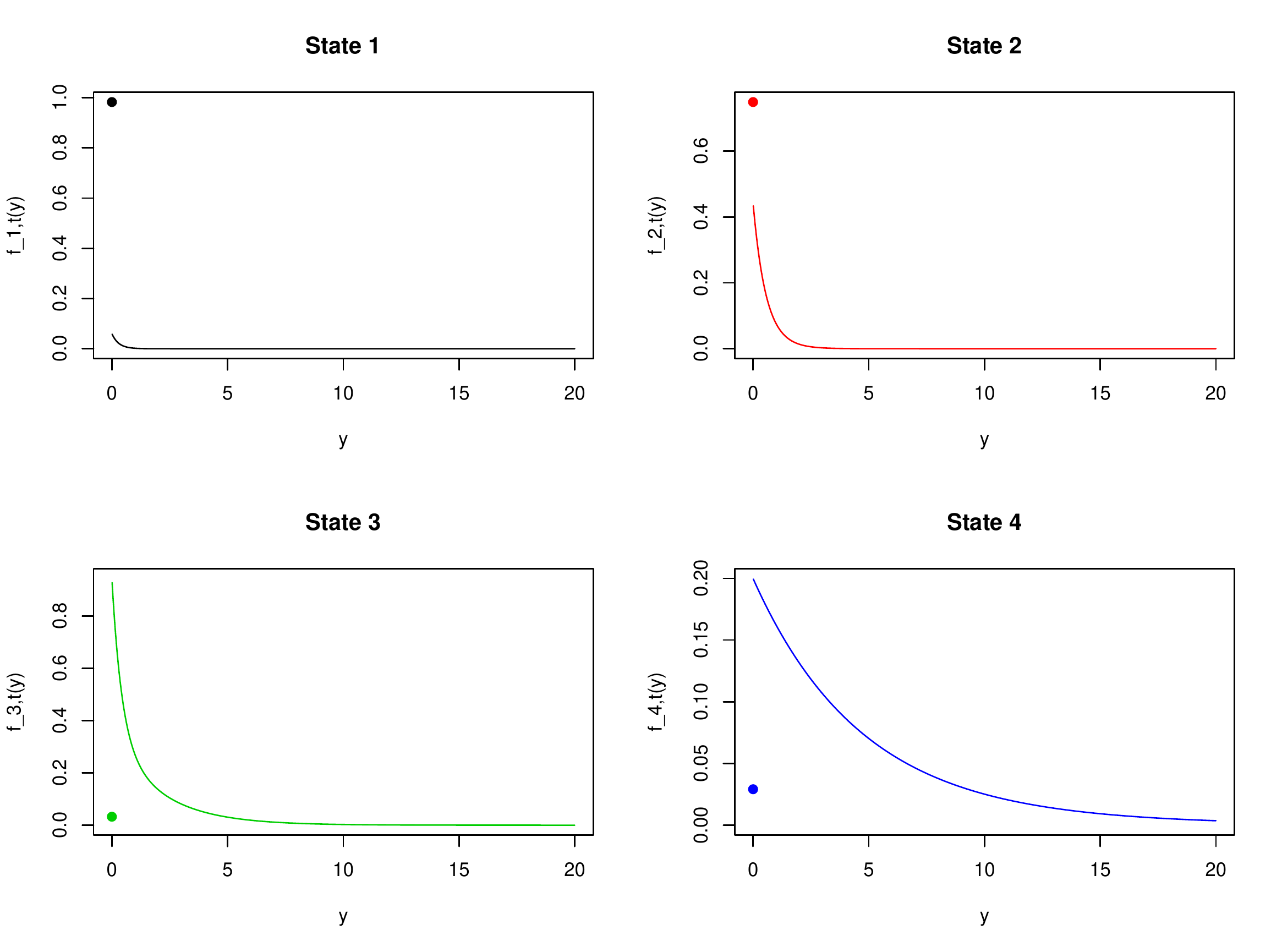}
\caption{Estimated emission densities}\label{dens}
\end{figure}

The physical interpretation of the four states is straightforward. State 1 is mostly a dry state: in this state, it rarely rains and when it does rain, the rainfalls amounts are small. On the opposite, state 4 is a rainy state, with heavy rainfalls. Between these two extremes, state 2 and state 3 are intermediate states with moderate precipitation amounts. However, they differ by their precipitations frequency, as state 2 is dry most of the time whereas state 3 is almost always rainy. Figure \ref{qhatrain} shows the transition probabilites between the four states as functions of time.

\begin{figure}[H]
\centering
\includegraphics[scale=0.5]{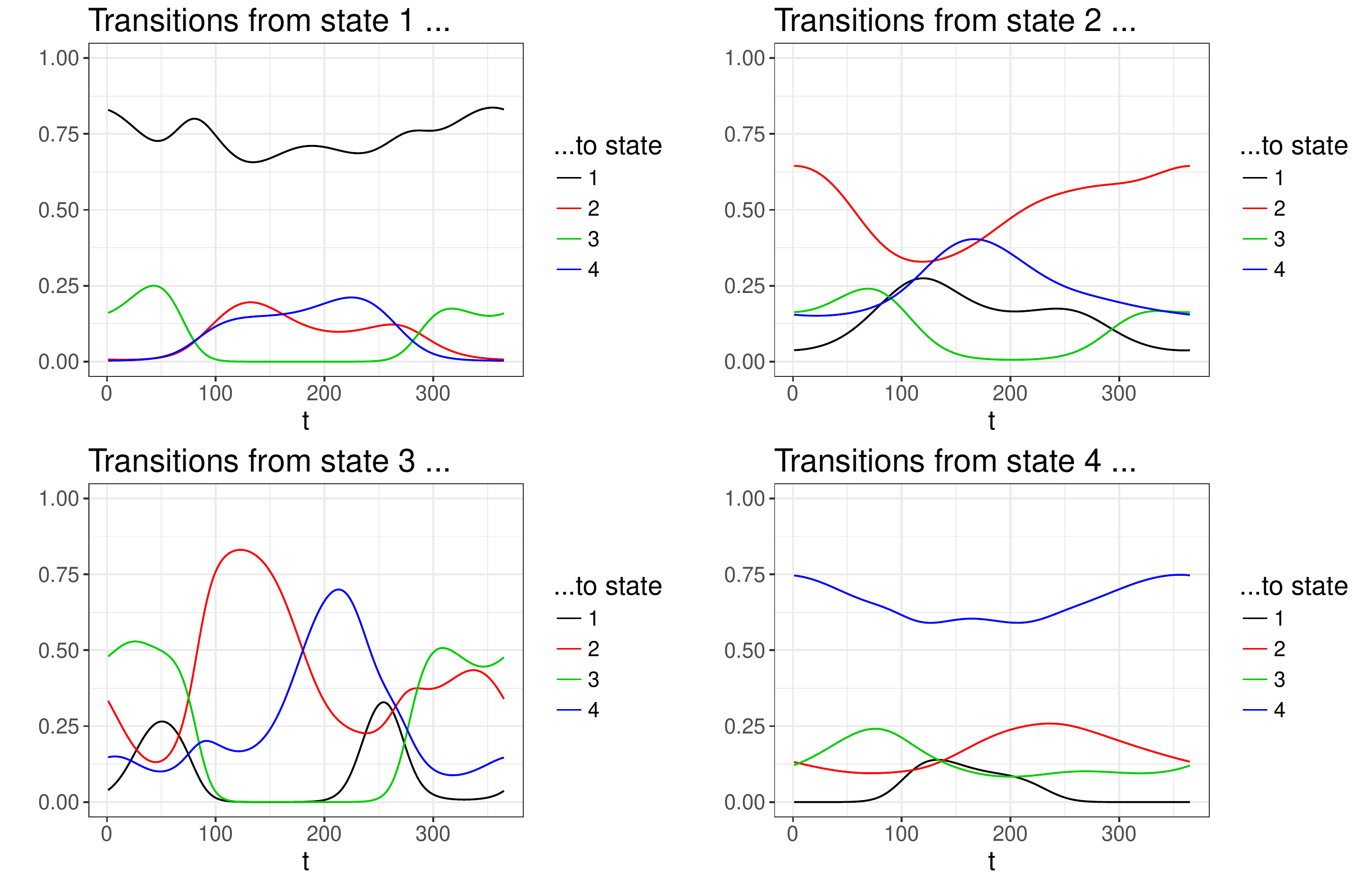}
\caption{Estimated transition probabilities}\label{qhatrain}
\end{figure}

It is also interesting to look at the relative frequencies of the four states (Figure \ref{stat}). These vary quite a lot throughout the year.

\begin{figure}[H]
\centering
\includegraphics[scale=0.4]{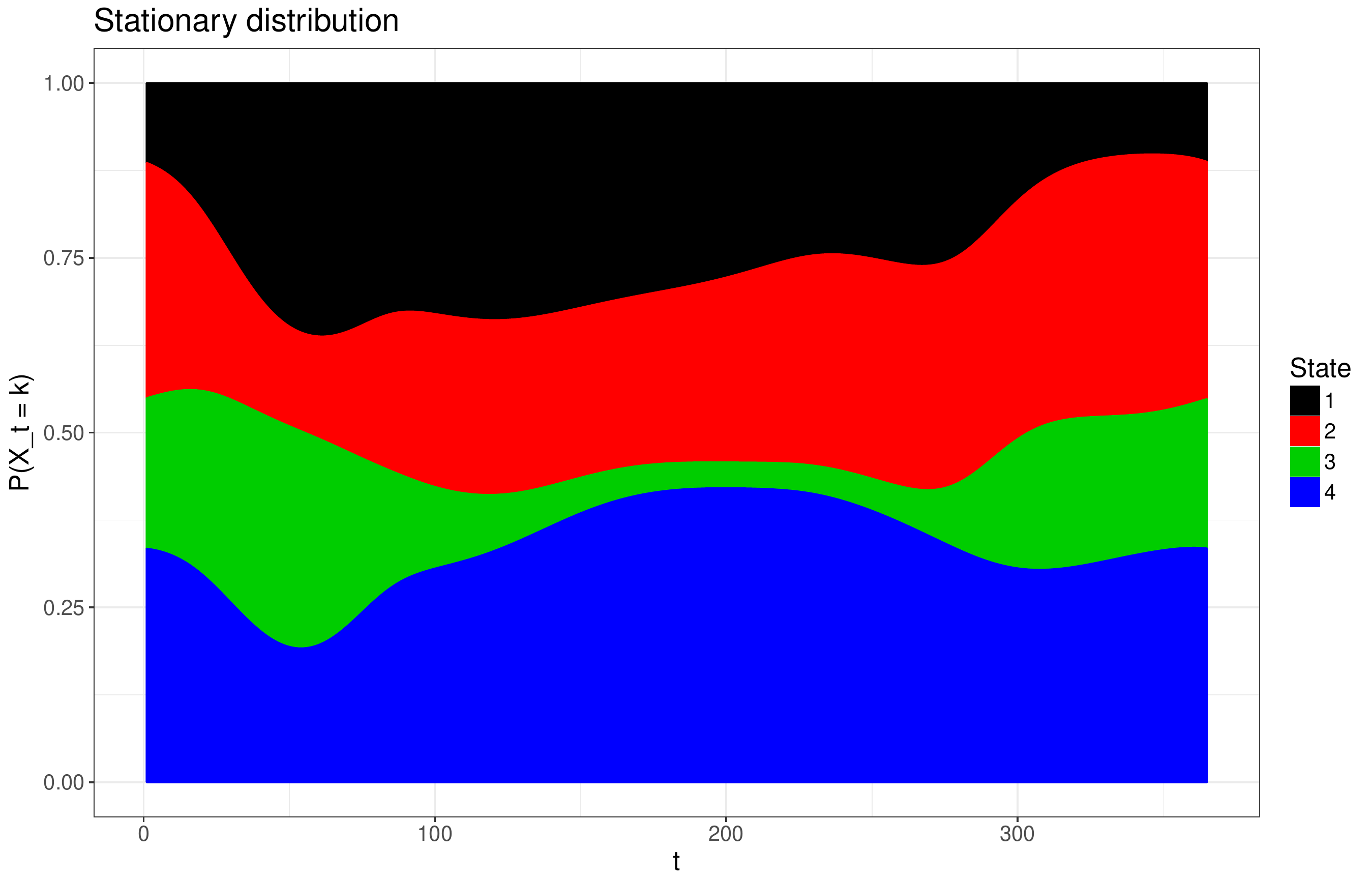}
\caption{Relative frequencies of states}\label{stat}
\end{figure}
  
In particular, we see that in summer, dry states 1 and 2 are less frequent whereas state 4 is the most visited state. It means that in summer, we observe either dry days or heavy rain. It is the opposite in winter, where rainfalls are more frequent but also lighter compared to the summer. These interpretations of the states are consistent with climatology. 

\paragraph{Validation}

As we wish to generate realistic simulations of daily rainfall amounts, i.e. simulations whose statistical properties mimic those of the real data, we evaluate the model by comparing the simulations produced by the model using the estimated parameters to the observed time series. To be specific, $1000$ independent simulations are produced, each of them having the same length as the observed series. To perform a simulation, we first simulate a Markov chain $\left(X_t^\mathrm{sim}\right)_{t\geq 1}$ with transition matrices $\hat{Q}(t)$. Then we simulate the observation process $\left(Y_t^\mathrm{sim}\right)_{t\geq 1}$ using the estimated densities $f_{X_t^\mathrm{sim},t}^{\hat{\theta}_Y}$. Several criteria can be considered to carry out the comparison: daily statistics (moments, quantiles, maxima, rainfall occurrence), overall distribution of precipitations, distribution of annual maximum, interannual variability, distribution of the length of dry and wet spells... The choice of the criteria mostly depends on the specific application of the model. Each of these statistics is computed from the simulations, which provides an approximation of the distribution of the quantity of interest under the law of the generator (in other words, we use \emph{parametric bootstrap}), hence a $95\%$ prediction interval.  Then this distribution is compared to the value of the same statistic computed using the data. Let us first compare the overall distributions of the real precipitation amounts and the simulated ones by looking at the quantile-quantile plot (see Figure \ref{qqplot}).

\begin{figure}[H]
\centering
\includegraphics[scale=0.4]{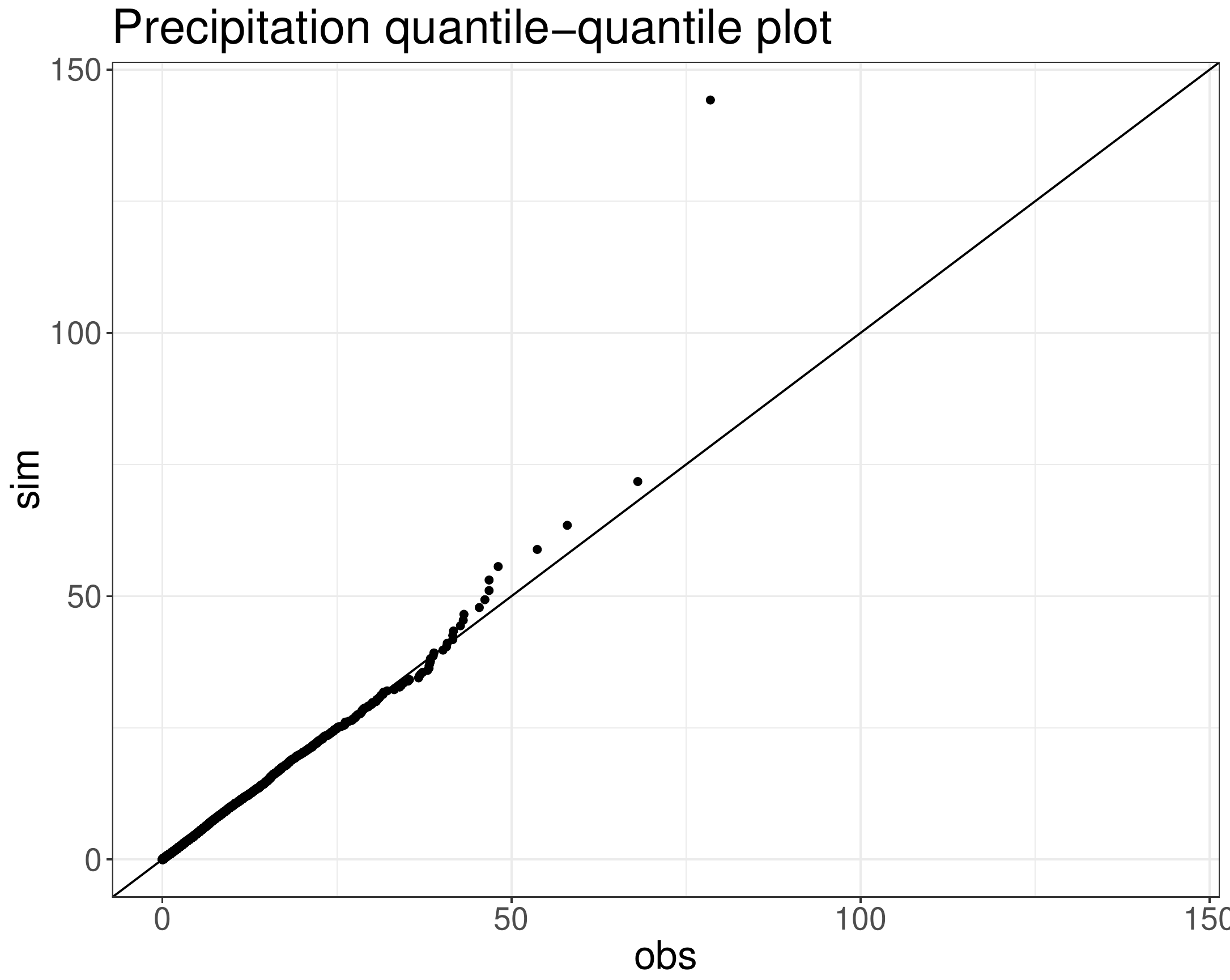}
\caption{Quantile-quantile plot}\label{qqplot}
\end{figure}

The match is correct, except in the upper tail of the distribution. The last point corresponds to the maximum of the simulated values, which is much larger than the maximum observed value. This should not be considered as a problem: a good weather generator should be able to (sometimes) generate values that are larger than those observed. We then focus on daily distributions. Figure \ref{moments} shows the results obtained for the first four daily moments and for the daily frequency of rainfall. It shows that these statistics are well reproduced by the model. Even though we did not introduce seasonal coefficients in emission densities, seasonalities appear both in the frequency of rainfall and the amounts. This is only due to the seasonality of the transition probabilities between the states.

\begin{figure}[H]
\centering
\includegraphics[scale=0.5]{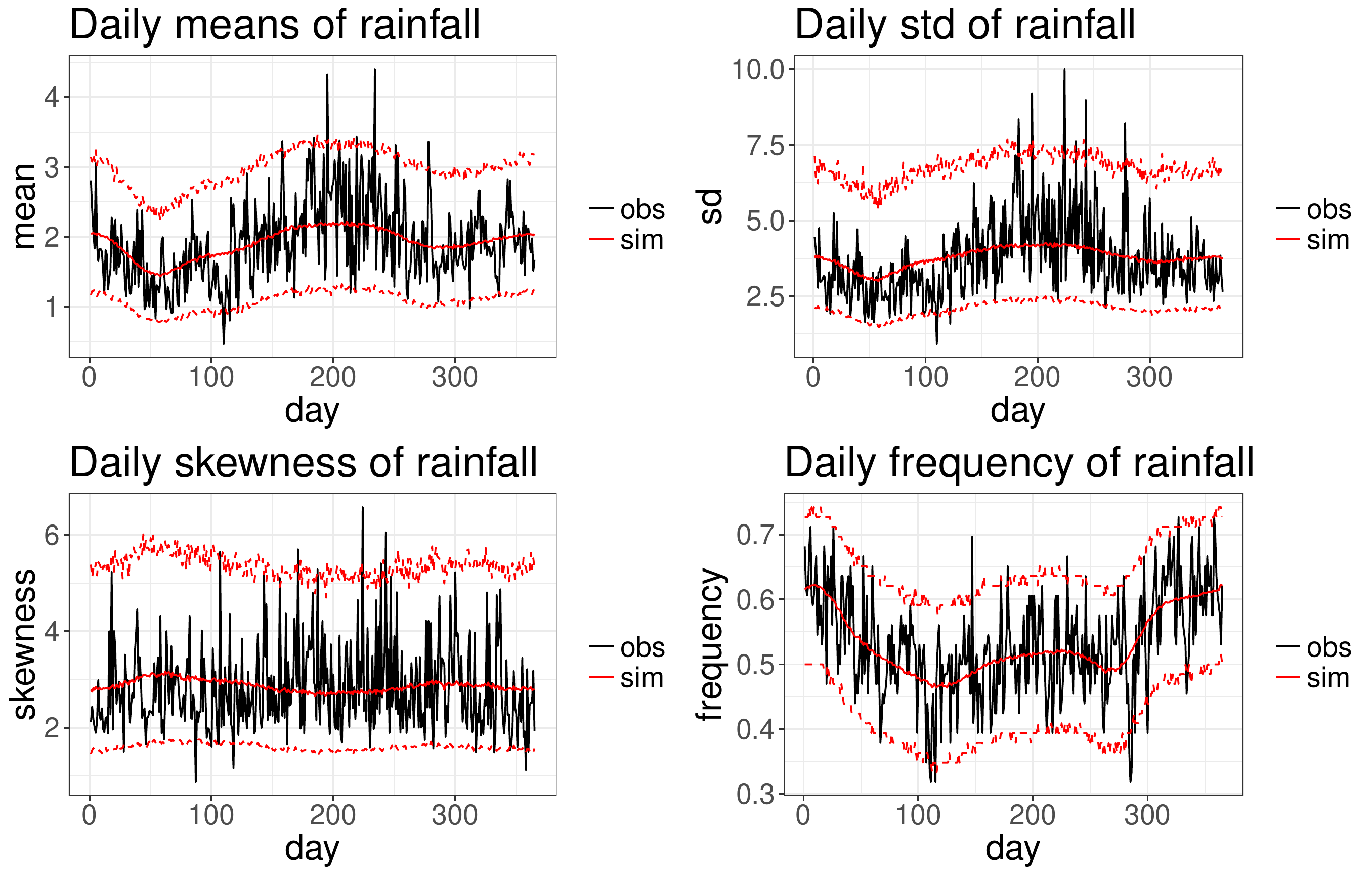}
\caption{Daily moments and frequency of precipitations. The black line relates to observations, the red solid line is the mean over all simulations, and the dashed lines depict an estimated $95\%$ prediction interval under the model.}\label{moments}
\end{figure}

The distribution of the duration of \emph{dry and wet spells} is another quantity of interest when one studies precipitation. A wet (resp. dry) spell is a set of consecutive rainy (resp. dry) days. This statistic provides a way to measure the time dependence of the occurrence process. The results are presented in Figure \ref{spells}. The dry spells are well modelled, whereas there is a slight underestimation of the frequency of 2-day wet spells while the single day events frequency is slightly overestimated.\\

\begin{figure}[H]
\centering
\includegraphics[scale=0.4]{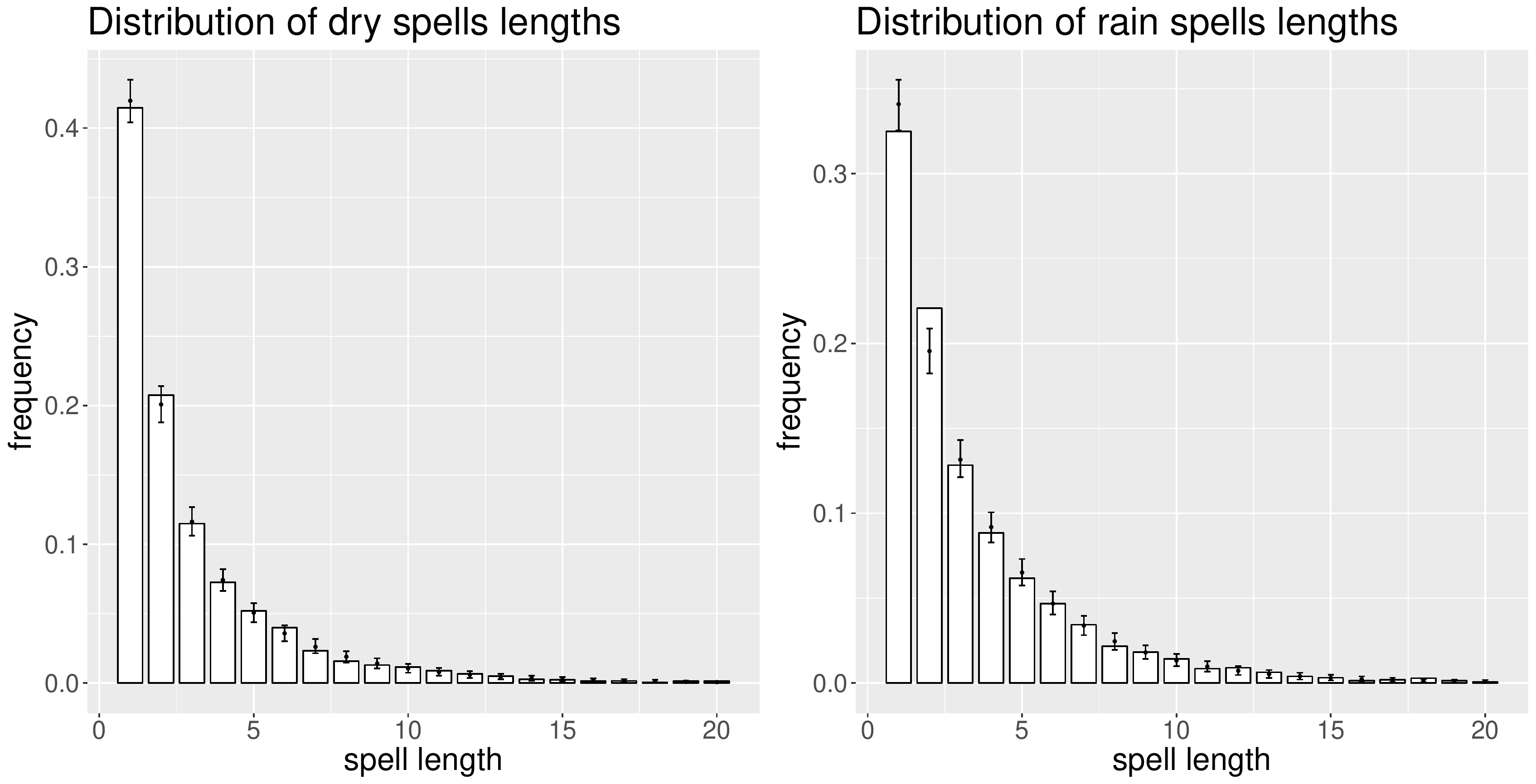}
\caption{Distribution of the lengths of dry (left plot) and wet (right plot) spells: observed (bars) versus simulated (error bars). The dots represent the means of the simulations}\label{spells}
\end{figure}

\section{Conclusion}

We introduced a variant of hidden Markov models called SHMM, adapted to data with a seasonal behaviour. In these models, the transition probabilities between the states are periodic, as well as the emission distributions. We gave sufficient conditions of identifiabiity for SHMM and we proved that under reasonable assumptions on the parameter space, the maximum likelihood estimator is strongly consistent, thus generalizing previous results on HMM. Two specific models for which those conditions are satisfied were given as examples. In order to compute the maximum likelihood estimator, we described the EM algorithm adapted to the framework of SHMM and we used it with simulated data to illustrates our consistency result. In the last section, we applied a SHMM with zero-inflated mixtures of exponential distributions as emission laws to precipitation data. We showed that such a model provides a good example of a stochastic weather generator, as the statistical properties of time series generated by the model are close to those of the observations.\\
In this paper, we considered the number of states of the hidden process as a known parameter. However, in most real world applications, it is unknown. This model selection problem has yet to be adressed. In many applications, the data exhibit trends (e.g. climate change) in addition to seasonalities. However, the techniques presented in this paper cannot be directly adapted to deal with trends, so that this case requires further investigation.

\section*{Acknowledgements}
The author would like to thank Yohann De Castro, \'Elisabeth Gassiat, Sylvain Le Corff and Luc Leh\'ericy from Universit\'e Paris-Sud for fruitful discussions and valuable suggestions. This work is supported by EDF. We are grateful to Thi-Thu-Huong Hoang and Sylvie Parey from EDF R\&D for providing this subject and for their useful advice. 

\bibliographystyle{plainnat}
\bibliography{bibli}

\end{document}